%% file: main.tex
\documentclass[11pt]{article}
\input{preamble}

\title{Relaxed Locally Correctable Codes in Computationally Bounded Channels}
\author{
Jeremiah Blocki\thanks{Department of Computer Science, Purdue University, West Lafayette, IN. 
Email: \email{jblocki@purdue.edu}.}
\and
Venkata Gandikota\thanks{Department of Computer Science, Johns Hopkins University, Baltimore, MD.
Email: \email{gv@jhu.edu}.}
\and
Elena Grigorescu\thanks{Department of Computer Science, Purdue University, West Lafayette, IN. 
Research supported in part by NSF CCF-1649515 and by a grant from the Purdue Research Foundation.
Email: \email{elena-g@purdue.edu}.}
\and
Samson Zhou\thanks{School of Informatics, Computing, and Engineering, Indiana University, Bloomington, IN. 
This work was done in part while at the Department of Computer Science, Purdue University, West Lafayette, IN. 
Research supported in part by NSF CCF-1649515. 
E-mail: \email{samsonzhou@gmail.com}}\\
}

\begin{document}
\maketitle
\input{abstract}
\input{intro}

\input{lcc1}
\input{lcc-dec}
\input{lcc-final}

\def\shortbib{0}
\bibliographystyle{alpha}
\bibliography{references,abbrev3,crypto}
\appendix
\input{relwork}
\input{prelims}
\input{app-lcc-dec}
\input{app-lcc-final}
\input{app-rldc}
\end{document}

%% file: preamble.tex
\usepackage{amsmath,amssymb,amsfonts}
\usepackage{fullpage}
\usepackage[backref]{hyperref}
\usepackage{color}
\usepackage{tikz}
\usetikzlibrary{decorations.pathreplacing}
\usepackage{setspace}
\usepackage[framemethod=tikz]{mdframed}
\usepackage{xspace}
\usepackage{mathtools}

\newenvironment{proof}{\noindent{\bf Proof : \ }}{\hfill$\Box$\par\medskip}

\newtheorem{theorem}{Theorem}[section]

\newtheorem{lemma}[theorem]{Lemma}

\newtheorem{definition}[theorem]{Definition}
\newtheorem{remark}[theorem]{Remark}
\newtheorem{claim}[theorem]{Claim}

\newtheorem{fact}[theorem]{Fact}

\newenvironment{proofof}[1]{\begin{trivlist} \item {\bf Proof
#1:~~}}
  {\qed\end{trivlist}}
\renewenvironment{proofof}[1]{\par\medskip\noindent{\bf Proof of #1: \ }}{\hfill$\Box$\par\medskip}
\newcommand{\namedref}[2]{\hyperref[#2]{#1~\ref*{#2}}}
\newcommand{\thmlab}[1]{\label{thm:#1}}
\newcommand{\thmref}[1]{\namedref{Theorem}{thm:#1}}
\newcommand{\lemlab}[1]{\label{lemm:#1}}
\newcommand{\lemref}[1]{\namedref{Lemma}{lemm:#1}}
\newcommand{\claimlab}[1]{\label{claim:#1}}
\newcommand{\claimref}[1]{\namedref{Claim}{claim:#1}}

\newcommand{\seclab}[1]{\label{sec:#1}}
\newcommand{\secref}[1]{\namedref{Section}{sec:#1}}
\newcommand{\applab}[1]{\label{app:#1}}
\newcommand{\appref}[1]{\namedref{Appendix}{app:#1}}

\newcommand{\deflab}[1]{\label{def:#1}}
\newcommand{\defref}[1]{\namedref{Definition}{def:#1}}

\newcommand{\ie}{\emph{i.e.}}
\newcommand{\error}{\kappa}
\newcommand{\metanode}[1]{\ensuremath{\mathsf{map}(#1)}}

\renewcommand{\d}{\delta}
\def \HAM    {\mdef{\mathsf{dist}}}

\def \LDC    {\mdef{\textsf{LDC}}}

\def \RLDC    {\mdef{\textsf{RLDC}}}

\def \ECC    {\mdef{\textsf{ECC}}}
\def \ECCD {\mdef{\textsf{ECCD}}}
\def \RLCC    {\mdef{\textsf{RLCC}}}

\def \CRLCC    {\mdef{\textsf{CRLCC}}}

\def \CRLCCs    {\mdef{\textsf{CRLCCs}}}

\def \CRLDC    {\mdef{\textsf{CRLDC}}}

\def \LDCs    {\mdef{\textsf{LDC}s}}

\def \CRHF    {\mdef{\textsf{CRHF}}}
\def \Gen    {\mdef{\mathsf{Gen}}}
\def \Goody    {\mdef{\mathsf{Good}}}
\def \Fool    {\mdef{\mathsf{Fool}}}
\def \Limit    {\mdef{\mathsf{Limit}}}
\def \GenH    {\mdef{\mathsf{GenH}}}
\def \LCCs    {\mdef{\textsf{LCC}s}}
\def \RLDCs    {\mdef{\textsf{RLDC}s}}
\def \RLCCs   {\mdef{\textsf{RLCC}s}}

\def \PPT    {\mdef{\textsf{PPT}}}
\def \PCPs    {\mdef{\textsf{PCP}s}}

\def \deg    {\mdef{\mathsf{deg}}}
\def \Green    {\mdef{\mathsf{IsGreen}}}
\def \Good    {\mdef{\mathsf{IsGood}}}

\def \GreenMeta    {\mdef{\mathsf{IsGreenMeta}}}
\def \GreenEdge    {\mdef{\mathsf{IsGreenEdge}}}
\def \DeltaExpansion {\mdef{\mathsf{IsLocalExpander}}}
\def \ReduceDegree    {\mdef{\mathsf{ReduceDegree}}}
\def \indeg    {\mdef{\mathsf{indeg}}}
\def \outdeg    {\mdef{\mathsf{outdeg}}}
\def \parents    {\mdef{\mathsf{parents}}}

\def \lab					{\mdef{\mathsf{lab}}}

\newcommand{\COMMENTED}[1]{{}}

\newcommand{\PPr}[1]{\ensuremath{\mathbf{Pr}\left[#1\right]}}

\renewcommand{\O}[1]{\ensuremath{\mathcal{O}\left(#1\right)}}

\newcommand{\eps}{\epsilon}

\newcommand{\mdef}[1]{{\ensuremath{#1}}\xspace}  
\newcommand{\myfunc}[1]{\mdef{\mathsf{#1}}}      

\DeclareMathOperator*{\polylog}{polylog}

\newcommand{\superscript}[1]{\ensuremath{^{\mbox{\tiny{\textit{#1}}}}}\xspace}
\def \th {\superscript{th}}     
\def \etal{{\it et~al.}\,}

\def \A        {\mdef{\mathcal{A}}}                  
\def \negl     {\mdef{\myfunc{negl}}}                
\def \polylog  {\mdef{\myfunc{polylog}\,}}             
\newcommand{\ceil}[1]{\mdef{\left\lceil#1\right\rceil}}               

\def \sec						      {\mdef \lambda}
\def \Enc							      {\mdef{\mathsf{Enc}}}
\def \Dec							      {\mdef{\mathsf{Dec}}}

\def \hashexperiment        {\mdef{\mathsf{Hash-coll}_{\A,\Pi}(\sec)}}

\newcommand{\ignore}[1]{}

\newif\ifnotes\notestrue 
\ifnotes
\newcommand{\elena}[1]{\textcolor{red}{{\bf (Elena:} {#1}{\bf ) }} \marginpar{\tiny\bf
             \begin{minipage}[t]{0.5in}
               \raggedright E:
                \end{minipage}}}
\newcommand{\samson}[1]{\textcolor{green}{{\bf (Samson:} {#1}{\bf ) }} \marginpar{\tiny\bf
             \begin{minipage}[t]{0.5in}
               \raggedright S:
            \end{minipage}}}
\newcommand{\gv}[1]{\textcolor{purple}{{\bf (GV:} {#1}{\bf ) }} \marginpar{\tiny\bf
             \begin{minipage}[t]{0.5in}
               \raggedright G:
            \end{minipage}}}  
\newcommand{\jeremiah}[1]{\textcolor{blue}{{\bf (Jeremiah:} {#1}{\bf ) }} \marginpar{\tiny\bf
             \begin{minipage}[t]{0.5in}
               \raggedright J:
            \end{minipage}}} 
\else
\newcommand{\elena}[1]{}
\newcommand{\samson}[1]{}
\newcommand{\gv}[1]{}
\newcommand{\jeremiah}[1]{}
\fi

\hypersetup{
     colorlinks   = true,
     citecolor    = blue,
		 linkcolor		= red
}

\makeatletter
\renewcommand*{\@fnsymbol}[1]{\textcolor{magenta}{\ensuremath{\ifcase#1\or *\or \dagger\or \ddagger\or
 \mathsection\or \triangledown\or \mathparagraph\or \|\or **\or \dagger\dagger
   \or \ddagger\ddagger \else\@ctrerr\fi}}}
\makeatother

\providecommand{\email}[1]{\href{mailto:#1}{\nolinkurl{#1}\xspace}}


%% file: abstract.tex

\begin{abstract}
Error-correcting codes that admit {\em local} decoding and correcting algorithms have been the focus of much recent research due to their numerous theoretical and practical applications.  An important goal is to obtain the best possible tradeoffs between the number of queries  the algorithm makes to its oracle (the {\em locality} of the task), and the amount of redundancy in the encoding (the {\em information rate}).

In Hamming's classical adversarial channel model, the current tradeoffs are dramatic, allowing either small locality, but superpolynomial blocklength, or small blocklength, but high locality.
However, in the computationally bounded, adversarial channel model, proposed by Lipton (STACS 1994), constructions of locally decodable codes suddenly exhibit small locality and small blocklength, but these constructions require strong trusted setup assumptions e.g., Ostrovsky, Pandey and Sahai (ICALP 2007) construct private locally decodable codes in the setting where the sender and receiver already share a symmetric key.

We study variants of locally decodable and locally correctable codes in computationally bounded,  adversarial channels, in a setting with no public-key or private-key cryptographic setup. The only setup assumption we require is the selection of the {\em public} parameters (seed) for a collision-resistant hash function. Specifically, we provide constructions of {\em relaxed locally correctable} and {\em relaxed locally decodable codes}  over the binary alphabet, with constant information rate, and poly-logarithmic locality. 

Our constructions, which compare favorably with their classical analogues in the computationally unbounded Hamming channel, crucially employ {\em collision-resistant hash functions} and {\em local expander graphs}, extending ideas from recent cryptographic constructions of memory-hard functions.
\end{abstract}

%% file: intro.tex

\section{Introduction}
\seclab{sec:intro}
\newcommand{\ra}{\rightarrow}
Classically, an error-correcting code is a tuple $(\Enc, \Dec)$ of encoding and decoding algorithms employed by a sender to encode messages, and by a receiver to decode them, after potential corruption by a noisy channel during transmission.
Specifically, the sender encodes a {\em message} $m$ of $k$ symbols from an alphabet $\Sigma$ into a {\em codeword} $c$ of block-length $n$ consisting of symbols  over the same alphabet, via $\Enc:\Sigma^k\ra \Sigma^n$. 
The receiver uses $\Dec: \Sigma^n\ra \Sigma^k$ to recover the message from a received word  $w\in \Sigma^n$, a corrupted version of some $\Enc(m)$.
Codes over the binary alphabet $\Sigma=\{0,1\}$ are preferred in practice. 
The quantities of interest in designing classical codes are the {\em information rate}, defined as $k/n$, and the {\em error rate}, which is the tolerable fraction of errors in the received word. 
Codes with both large information rate and large error rate are most desirable. 

In modern  uses of error-correcting codes, one may only need to recover small portions of the message, such as a single bit. 
In such settings, the decoder may not need to read the entire received word $w\in \Sigma^n$, but only read a few bits of it. 
Given an index $i\in [n]$ and oracle access to $w$, a local decoder must make only $q=o(n)$ queries into $w$, and output the bit $m_i$. 
Codes that admit such fast decoders are called {\em  locally decodable codes} (\LDCs) \cite{KatzT00,SudanTV99}. 
The parameter $q$ is called the {\em locality} of the decoder. 
A related notion is that of {\em locally correctable codes} (\LCCs). \LCCs are codes for which the local decoder with oracle access to $w$ must output bits of the codeword $c$, instead of bits of the message $m$. \LDCs and \LCCs have  widespread applications in many areas of theoretical computer science, including private information retrieval, probabilistically checkable proofs, self-correction,  fault-tolerant circuits,  hardness amplification, and data structures (e.g., \cite{BabaiFLS91,LundFKN92,BlumLR93,BlumK95,ChorKGS98,ChenGW13,AndoniLRW17}). 
See surveys \cite{Tre04-survey,Gasarch04}. 
However, constructions of such codes suffer from apparently irreconcilable tension between locality and rate: existing codes with constant locality have slightly subexponential blocklength \cite{Yekhanin08,Efremenko12,DvirGY11}, and codes of linear blocklength have slightly subpolynomial query complexity \cite{KoppartyMRS17}. 
See surveys by Yekhanin \cite{Yekhanin12} and by Kopparty and Saraf \cite{KoppartyS16}.

Ben-Sasson \etal\cite{Ben-SassonGHSV06} propose the notion of {\em relaxed locally decodable codes } (\RLDCs) that remedies  
the dramatic tradeoffs of classical \LDCs. In this notion the decoding algorithm is allowed to output $\bot$ sometimes, to signal that it does not know the correct value; however, it should not output an incorrect value too often. 
More formally, given $i\in [k]$ and oracle access to the received word $w$ assumed to be relatively close  to some codeword $c=\Enc(m) \in \Sigma^n$, the local decoder (1) outputs $m_i$ if $w=c$; (2) outputs either $m_i$ or $\bot$ with probability $2/3$,  otherwise;  and 
(3) the set of indices $i$ such that the decoder outputs $m_i$ (the correct value) with probability $2/3$ is of size $>\rho \cdot k$ for some constant $\rho>0$.
The relaxed definition allows them to achieve \RLDCs with constant query complexity and  blocklength $n=k^{1+\eps}$. 

Recently, Gur \etal\cite{GurRR18} introduce the analogous  notion of  {\em relaxed locally correctable codes} (\RLCCs). 
In particular, upon receiving a word $w\in \Sigma^n$ assumed to be close to some codeword $c$, the decoder: (1) outputs $c_i$ if $w=c$; (2) outputs either $c_i$ or $\bot$ with probability $2/3$,  otherwise; and (3) the set of indices $i$ such that the decoder outputs $c_i$ with probability $2/3$ is of size $\rho \cdot k$, for some $\rho>0$.
In fact, \cite{GurRR18} omits condition (3) in their definition, since the first two conditions imply the 3rd, for codes with constant locality that can withstand a constant fraction of error \cite{Ben-SassonGHSV06}. 
The reduction from \cite{Ben-SassonGHSV06}, however, does not maintain the asymptotic error rate, and in particular, in the  non-constant query complexity regime, the error rate becomes subconstant. 
Since our results work in the $\omega(1)$-query regime, we will build codes that achieve the 3rd condition as well (for constant error rate).
The results in \cite{GurRR18} obtain  significantly better parameters for \RLCCs than for classical \LCCs; namely, they construct \RLCCs with constant query complexity, polynomial block length, and constant error rate, and  \RLCCs with quasipolynomial query complexity, linear blocklength (constant rate), with the caveat that the error rate is subconstant. 
These results immediately extend to \RLDCs, since their codes are {\em systematic}, meaning that the initial part of the encoding consists of the message itself.

\paragraph{Computationally Bounded, Adversarial Channels} In this work we study \RLDCs and \RLCCs in the more restricted, yet natural, {\em computationally bounded adversarial channel}, introduced by Lipton \cite{Lipton94}. 
All the above constructions of local codes assume a channel that may introduce a bounded number of adversarial errors, and the channel has as much time as it needs to decide what positions to corrupt. 
This setting corresponds to  Hamming's error model, in which codes should be resilient to any possible error pattern, as long as the number of corruptions is bounded. 
Hence, codes built for the Hamming channel are safe for data transmission, but the drastic requirements lead to coding limitations. 
By contrast, in his foundational work, Shannon proposes a weaker, probabilistic channel in which each symbol is corrupted  independently, with some fixed probability. 

In  \cite{Lipton94}, Lipton argues that many reasonable channels stand somewhere in between these two extremes, and one may assume that in reality adversaries are computationally bounded, and can be modeled as {\em polynomial time probabilistic} (\PPT) algorithms. 
Variants of this model have been initially studied for classical error-correcting codes \cite{Lipton94,GopalanLD04,MicaliPSW05,GuruswamiS16, ShaltielS16} to show better error rate capabilities than in the Hamming model. See also Appendix \ref{sec:relwork} for further details on related work.
 
Ostrovsky \etal \cite{OstrovskyPS07} introduce the study of ``private'' locally decodable codes against computationally bounded channels, and build codes with constant information and error rates over the binary alphabet, which can correct every message bit using only a small (superconstant) number of queries to the corrupted word. 
Their results assume the existence of one-way functions, and require that the encoding and decoding algorithms share a secret key that is not known to the channel. 
Hemenway and Ostrovky \cite{HemenwayO08} and Hemenway \etal \cite{HemenwayOSW11} construct public-key \LDCs assuming the existence of $\Phi$-hiding schemes \cite{CachinMS99} and of IND-CPA secure cryptosystems, respectively.
   
 By contrast, our constructions of \RLDCs and \RLCCs do not require the sender and receiver to exchange cryptographic keys. 
Instead our constructions are based on the existence of collision-resilient hash functions, a standard cryptographic assumption. 
Because the parameters of a collision-resistant hash function are public, {\em any} party (sender/receiver/attacker) is able to compute it.

\section{Our Contributions}

We start by defining the version of relaxed locally correctable codes that is relevant for our model. 
We remark that our codes are systematic as well, and therefore the results also apply to the  \RLDCs analogue.

Since these codes interact with an adversarial channel, their strength  is not only measured in their error correction and locality capabilities (as is the case for \RLCCs/\RLDCs in the Hamming channel), but also in the security they provide against the channel.
We present these codes while describing how they interact with the channel, in order to make the analogy with the classical setting. 
We use the notation $\Enc$ and $\Dec$ to denote encoding and decoding  algorithms.

\begin{definition} A {\em local code} is a tuple $(\Gen, \Enc, \Dec)$ of probabilistic algorithms such that  
\begin{itemize}
\item $\Gen(1^{\sec})$ takes as input security parameter $\sec$ and generates a public seed $s \in \{0,1\}^{*}$. 
This public seed $s$ is {\em fixed} once and for all.  
\item $\Enc$ takes as input the public seed $s$ and a message $x \in \Sigma^k$ and outputs a codeword $c =\Enc(s,x)$ with $c \in \Sigma^n$.  
\item $\Dec$ takes as input the public seed $s$, an index $i \in [n]$, and is given oracle access to a word $w \in \Sigma^n$. 
$\Dec^w(s,i)$ outputs a symbol $b\in \Sigma$ (which is supposed to be the value at position $i$ of the closest codeword to $w$).
\end{itemize}
We say that the (information) {\em rate} of the code  $(\Gen, \Enc, \Dec)$ is $k/n$. 
We say that the code is {\em efficient} if $\Gen,\Enc,\Dec$ are all probabilistic polynomial time (PPT) algorithms.
\end{definition}

\begin{definition}
\deflab{def:crlcc:channel}
 A {\em computational adversarial channel} $\A$ with error rate $\tau$ is an algorithm that interacts with a local code $(\Gen, \Enc, \Dec)$  in rounds, as follows. In each round of the execution, given a security parameter $\sec$, 
\begin{enumerate}
\item  Generate $s\leftarrow\Gen(1^\sec)$; $s$ is public, so $\Enc$, $\Dec$, and $\A$ have access to $s$
\item The channel $\A$ on input $s$ hands a message $x$ to the sender.
\item The sender computes $c=\Enc(s, x)$ and hands it back to the channel (in fact the channel can compute $c$ without this interaction).
\item The channel $\A$ corrupts at most $\tau n$ entries of $c$ to obtain a word $w\in \Sigma^n$ and selects a challenge index $i \in [n]$;  $w$ is given to the receiver's \Dec with query access along with the challenge index $i$.
\item The receiver outputs $b\leftarrow \Dec^w(s, i)$.
\item We define $\A(s)$'s  {\em probability of fooling} $\Dec$ on this round to be $p_{\A, s}=\Pr[b\not \in \{\bot, c_i \}],$ where the probability is taken only over the randomness of the $\Dec^w(s, i)$. 
We say that $\A(s)$ is $\gamma$-successful {\em at fooling} $\Dec$ if $p_{\A, s}>\gamma$. 
We say that $\A(s)$ is $\rho$-successful {\em at limiting} $\Dec$ if $\left|\Goody_{\A,s}\right| < \rho \cdot n$, where $\Goody_{\A,s} \subseteq [n]$ is the set of indices $j$ such that $\Pr[\Dec^w(s,j) = c_j ] > \frac{2}{3}$. 
We use $\Fool_{\A,s}(\gamma,\tau,\sec)$ (resp. $\Limit_{\A,s}(\rho,\tau,\sec)$) to denote the event that the attacker was $\gamma$-successful at fooling $\Dec$ (resp. $\rho$-successful at limiting $\Dec$) on this round. 
\end{enumerate}
\end{definition}
\noindent
We now define our secure  \RLCC codes against computational adversarial channels.

\begin{definition}[(Computational) Relaxed Locally Correctable Codes (CRLCC)] A local code $(\Gen, \Enc, \Dec)$ is a $(q, \tau, \rho, \gamma(\cdot),\mu(\cdot))$-\CRLCC ~ {\em against a class} $\mathbb{A}$ of adversaries if $\Dec^w$ makes at most $q$ queries to $w$ and satisfies the following:
\begin{enumerate}
\item\label{weak1} For all public seeds $s$ if $w \leftarrow \Enc(s, x)$ then $\Dec^w(s,i)$ outputs $b=(\Enc(s, x))_i$. 
\item\label{weak2} For all $\A \in \mathbb{A}$ we have $\Pr[\Fool_{\A,s}(\gamma(\sec),\tau,\sec)] \leq \mu(\sec)$, where the randomness is taken over the selection of $s \leftarrow \Gen(1^\sec)$ as well as $\A$'s random coins. 
\item \label{strong} For all $\A \in \mathbb{A}$ we have $\Pr[\Limit_{\A,s}(\rho,\tau,\sec)] \leq \mu(\sec)$, where the randomness is taken over the selection of $s \leftarrow \Gen(1^\sec)$ as well as $\A$'s random coins. 
\end{enumerate}

  When $\mu(\sec)=0$, $\gamma(\sec)=\frac{1}{3}$ is a constant and $\mathbb{A}$ is the set of all (computationally unbounded) channels we say that the code is a $(q, \tau, \rho, \gamma)$-\RLCC. When $\mu(\cdot)$ is a negligible function {\em and} $\mathbb{A}$ is restricted to the set of all probabilistic polynomial time (\PPT) attackers we say that the code is a $(q, \tau, \rho, \gamma)$-\CRLCC (computational relaxed locally correctable code). 
  
We say that a code that satisfies conditions \ref{weak1} and \ref{weak2} is a {\em Weak \CRLCC}, while a code satisfying conditions \ref{weak1}, \ref{weak2} and \ref{strong} is a {\em Strong \CRLCC} code.
\end{definition}

We construct Weak and Strong \CRLCCs against \PPT adversaries, under the assumption that  {\em Collision-Resistant Hash Functions} (\CRHF) exist. 
Briefly, a \CRHF function is a pair $ (\GenH,H)$ of probabilistic polynomial time (\PPT) algorithms, where $\GenH$  takes as input a security parameter $\sec$ and outputs a public seed $s \in \{0,1\}^*$; the function $H: \{0,1\}^*\times \Sigma^*\ra \Sigma^{\ell(\sec)}$, takes as input the seed $s$ and a long enough input that is hashed into a string of length $\ell(\sec)$. 
We note that $H$ is deterministic upon fixing $s$. 
The value $\ell(\sec)$ is the {\em length} of the hash function. $(\GenH, H)$ is said to be collision-resistant if for all \PPT adversaries that take as input the seed $s$ generated by $\Gen(1^{\sec})$, the probability that they produce a collision pair $(x, x')$, i.e. such that $H(s, x)=H(s, x')$ and $x\ne x'$, is negligible in $1^{\sec}$.

\begin{theorem}
\thmlab{thm:rlcc}
Assuming the existence of a collision-resistant hash function $ (\GenH,H)$   with length $\ell(\lambda)$, 
there exists a constant $0<\tau' < 1$ and negligible functions $\mu(\cdot),\gamma(\cdot)$  such that for all $\tau \leq \tau'$ there is a constant $0< \rho(\tau)<1$ , such there exists a constant rate $(\ell(\lambda)\cdot {\polylog n}, \tau, \rho(\tau), \gamma(\cdot),\mu(\cdot))$-Weak \CRLCC of blocklength $n$ over the binary alphabet. In particular, if $\ell(\lambda) =\polylog \sec$ and $\sec \in \Theta(n)$ then the code is a $({\polylog n}, \tau, \rho, \gamma,\mu(\cdot))$-Weak \CRLCC. 
\end{theorem}

We note that in the above constructions the codes withstand a constant error rate. 
The classical \RLCCs of \cite{GurRR18} achieve $(\log n)^{\O{\log \log n}}$ query complexity, constant information rate, but subconstant error rate, in the Hamming channel. 
In order to obtain Strong \CRLCC we need further technical insights, that builds upon the Weak \CRLCCs. 
Our Strong \CRLCCs have the same parameters, with only a $\polylog n$ loss in query complexity.

\begin{theorem} \thmlab{strongRLCC}
 Assuming the existence of a collision-resistant hash function $ (\GenH,H)$   with length $\ell(\lambda)$, there exists a constant $0<\tau' < 1$ and negligible functions $\mu(\cdot),\gamma(\cdot)$ such that for all $\tau \leq \tau'$ there exist constants $0< r(\tau), \rho(\tau)<1$ such that there exists a $(\ell(\lambda)\cdot{\polylog n}, \tau, \rho(\tau), \gamma(\cdot),\mu(\cdot))$-Strong \CRLCC of blocklength $n$ over the binary alphabet with rate $r(\tau)$ where $r(.)$ and $\rho(.)$ have the property that $\lim_{\tau \rightarrow 0} r(\tau) = 1$ and $\lim_{\tau \rightarrow 0} \rho(\tau) = 1$. In particular, if $\ell(\lambda) =\polylog \sec$ and $\sec \in \Theta(n)$ then the code is a $({\polylog n}, \tau, \rho(\tau), \gamma(\cdot),\mu(\cdot))$-Strong \CRLCC.
\end{theorem}

\begin{remark}
We remark that any Weak \CRLCC can be trivially converted into a Strong \CRLCC e.g., by repeating the last bit of the codeword $cn$ times to obtain a new codeword of length $(c+1)n$. If we select the constant $c>0$ large enough then we can guarantee that, say $\tau (c+1)n \leq 4cn/10$, so that a majority of these repeated bits in the codeword are not corrupted and can be corrected with a simply majority vote. Thus, we can artificially ensure that $\rho(\tau) > \frac{c}{c+1}$ since at least $cn$ bits (out of $(c+1)n$ bits) of the codeword can be locally recovered by majority vote. However, to ensure that $\rho(\tau) > \frac{2}{3}$ we would need to set $c=2$ which {\em necessarily}  reduces the rate  to $r(\tau) < \frac{1}{c+1} = \frac{1}{3}$. By contrast, our  Strong \CRLCC has rate $r(\tau)$ which approaches $1$ as $\tau$ approaches $0$. Furthermore, we have the property that $\rho(\tau)$ approaches $1$ as $\tau$ approaches $0$ which means that {\em almost all bits} of the original codeword can be locally recovered. We also remark that our constructions are systematic and can be tweaked to yield the existence of a Strong \CRLDC with the same parameters. See additional discussion in \appref{crldc}.
\end{remark}

\input{TechnicalOverview}

%% file: TechnicalOverview.tex
\subsection{Technical Ingredients}
At a technical level our construction uses two main building blocks: local expander graphs and collision resistant hash functions. 

\paragraph*{Local Expander Graphs} Intuitively, given a graph $G=(V, E)$ and distinguished subsets $A,B \subseteq V$ of nodes such that $A$ and $B$ are disjoint and $|A|\leq |B|$ we say that the pair $(A,B)$ contains a $\delta$-expander if for all $X\subseteq A$ and $Y \subseteq B$ with $|X| > \delta |B|$ and $ |Y| > \delta|B|$ there is an edge connecting $X$ and $Y$ \emph{i.e.}, $(X \times Y) \cap E \neq \emptyset$. We say that a DAG $G$ is a $\delta$-local expander around a node $v$ if  for {\em any} radius $r>0$ and {\em any} node $v \geq 2r$ the pair $A = \{v-2r+1,\ldots, v-r\}$ and $B = \{v-r+1,\ldots, v\}$  contain a $\delta$-expander {\em and} the pair $C=\{v,\ldots,v+r-1\}$ and $D=\{v+r,\ldots,v+2r-1\}$ contain a $\delta$-expander. When this property holds for {\em every} node $v \in V(G)= \{1,\ldots,n\}$ we simply say that the DAG $G$ is a $\delta$-local expander. For any constant $\delta>0$ it is possible to (explicitly) construct a $\delta$-local expander with the property that $\indeg(G) \in \O{\log n}$ and $\outdeg(G) \in \O{\log n}$ i.e. no node  has more than $\O{\log n}$ incoming or outgoing edges~\cite{ErdosGS75,EC:AlwBloPie18}. 

Local expanders have several nice properties that have been recently exploited in the design and analysis of secure data-independent memory hard functions (iMHFs)~\cite{CCS:AlwBloHar17,EC:AlwBloPie17,TCC:BloZho17,EC:AlwBloPie18}. 
For example, these graphs are maximally depth-robust~\cite{EC:AlwBloPie18}. Even if we delete a {\em large} number of nodes $S \subseteq V$ the graph still contains a directed path of length $n-(1+\epsilon)\left|S\right|$ for some small constant $\epsilon$ dependent on $\delta$ --- the constant $\epsilon$ can approach $0$ as  $\delta$ approaches $0$~\cite{EC:AlwBloPie18}. 
More specifically, if we delete a large number of nodes $S \subseteq V$ at least $n-(1+\epsilon)\left|S\right|$  of the nodes have the property that they are $\alpha$-good (with $\epsilon = \left(\frac{2-\alpha}{\alpha}\right)$)  with respect to the deleted set $S$ and {\em any} pair of $\alpha$-good nodes $u$ and $v$ are connected by a directed path (provided that $\delta$ is sufficiently small) --- a node $v$ is $\alpha$-good with respect to $S$ if for any radius $r < v$ we have at most $\alpha r $ nodes in $S \cap [v-r+1,v]$ {\em and} for any radius $r \leq n-v+1$ we have at most  $\alpha r $ nodes in $S \cap[v,v+r-1]$. For more formal statements we point the reader to~\appref{app:prelims}.
In the context of memory hard functions each node in the graph corresponds to an intermediate data value generated during computation of the iMHFs and edges represent data-dependencies. It is known that an iMHF has high cumulative memory complexity~\cite{STOC:AlwSer15} if and only if the underlying DAG $G$  is depth-robust~\cite{C:AlwBlo16,EC:AlwBloPie17}~\footnote{Oversimplifying a bit, if an attacker attempts to reduce memory usage during computation then the attacker's running time will increase dramatically since, by depth-robustness, there will be a long chain of dependent data-values that that the attacker needs to recompute in the near future.}.

Suppose that each node is colored red or green and that we are only allowed to query each node to obtain its color. 
If we let $S$ denote the set of red nodes then we can develop an {\em efficient} randomized testing algorithm to check if a node $v$ is $\alpha$-good or not. 
The tester will make $\O{\polylog n}$ queries and, with high probability, will {\em accept} any node $v$ that is $\alpha_1$-good and will {\em reject} any node $w$ that is {\em not} $\alpha_2$-good for {\em any} constants $\alpha_2 > 4\cdot\alpha_1$. 
Intuitively, for each $r \in \{2^1,2^2,\ldots,2^{\log n}\}$ the tester will sample $\O{\polylog n}$ nodes in the intervals $[v, r-1]$ and $[v-r+1,v]$ to make sure that the fraction of red nodes is at most $2 \alpha_1$. 
If the tester determines that a node $v$ is {\em at least} $\alpha_2$-good for an appropriately small constant $\alpha_2$ then we can be (almost) certain that the long green path which contains {\em all} $n-(1+\epsilon)|S|$ of the  $\alpha_2$-good nodes also includes $v$. Furthermore, if $v < 3n/4$ then $v$ has at least $n/4- (1+\epsilon)|S|$ descendants in this directed green path. Intuitively, in our construction any such node $v$ must correspond to an {\em uncorrupted} portion of the codeword.

\paragraph{Collision Resistant Hash Functions} Our constructions employ {\em collision resistant hash functions} as a building block. 
While most of the recent progress on memory hard functions in cryptography combines local expanders (depth-robust graphs) with random oracles (e.g., see \cite{STOC:AlwSer15,TCC:AlwTac17,EC:AlwBloPie18}), we stress that we do not need to work in the random oracle model\footnote{The random oracle model is a source of some controversy among cryptographers~\cite{JC:KobMen07,EC:Menezes12,EPRINT:Goldreich06b,EPRINT:KobMen15} with some arguing that the framework can be used to develop cryptographic protocols that are {\em efficient and secure}~\cite{CCS:BelRog93}, while others argue that the methodology is flawed e.g.~\cite{AC:BerLan13,EPRINT:KobMen15}.}. 
Indeed, our constructions only assume the existence of collision resistant hash functions. 

\paragraph{Weak \CRLCCs} We first explain our construction of Weak \CRLCCs, which involves labeling a $\delta$-local expander with $k$ nodes. 
In particular, given an input word $x= (x_1\circ \ldots \circ x_k)$ (broken up into bit strings of length ${\ell(\sec)}$) and a $k$ node local expander graph $G$, the label of node $v$ is computed as $\ell_{v,s} = H\left(s, x_v \circ \ell_{v_1,s}\circ\ldots\circ \ell_{v_d,s}\right) \in\{0,1\}^{\ell(\sec)},$ where $\ell_{v_1,s}, \ldots ,  \ell_{v_d,s}$ are the labels of the parents $v_1,\ldots, v_d$ of node $v$, and $\circ$ denotes string concatenation.  
When $\ell(\sec) \in \O{\polylog \sec}$ we will select $\sec \in \O{n}$ to ensure that $\ell(\sec) \in \O{\polylog n}$.
We use the notation $\Enc$ and $\Dec$ for the encoding and decoding of our construction, while we use $\ECC$ and $\ECCD$ to denote the efficient encoding and decoding algorithms for a good binary code with constant rate and relative distance (e.g., \cite{Justesen72,SipserS96}) which can decode efficiently from some constant fractions of errors.

We first apply 
$\ECC$ to $x_1,\ldots, x_k$ to obtain codewords $c_1,\ldots,c_k \in \{0,1\}^{\O{\ell(\sec)}}$ where $c_i = \ECC(x_i)$. Also, for $v \in [k]$ we let $c_{v+k} = \ECC(\ell_{v,s})$ which is the encoding of the label corresponding to the node $v$ in $G$. 
The final output is $c=\left(c_1 \circ \ldots \circ c_{2k-1} \circ c_{2k} \circ c_{2k+1} \circ \ldots \circ c_{3k}\right)$ where $c_{2k+1}= \ldots = c_{3k}=c_{2k}$ consists of $k$ copies of the last codeword $c_{2k}$. 
The final word is an $n$ bit message with $n = \O{k\ell\left( \sec\right)}$. 
By repeating this last codeword $k$ times we ensure that it is {\em not possible} for the attacker to irreparably corrupt the final label $\ell_{k,s}$.  
 
Given a (possibly corrupted) codeword $c'$ produced by a \PPT attacker $\A$ 
we let $x' = (x_1' \circ \ldots \circ x_k')$ with $x_i' = \ECCD(c_i')$ (possibly $\bot$) and we let $\ell_{v,s}' = \ECCD(c_{v+k}')$ for $v\in[k]$ and $\ell_{k,s,j}' = \ECCD(c_{2k+j}')$ for each $j \in [k]$. 
We say that a node $v$ is {\em green} if it is locally consistent i.e.,
\[\ell_{v,s}' = H\left(s, x_v' \circ \ell_{v_1,s}' \circ \ldots \circ  \ell_{v_d,s}'\right),\]
otherwise, we say that the node is {\em red}. 
We first show that if a green node has the correct label $\ell_{v,s}' = \ell_{v,s}$ then it {\em must} be the case that $x_v' = x_v$ and $\ell_{v_i,s}' = \ell_{v_i,s}$ for each $i \leq d$ --- otherwise $\A$ would have found a hash collision! 
If a graph contains too many red nodes then this is easily detectable by random sampling and our {\em weak} local decoder is allowed to output $\bot$ if it detects {\em any} red nodes. 
Our local decoder will first obtain the final label $\ell_{k,s}$ by random sampling some labels from $\ell_{k,s,1}', \ldots, \ell_{k,s,k}'$ and checking to make sure each of these labels is equal to $\ell_{k,s}'$. 
If this check passes then with high probability we must have $\ell_{k,s}' = \ell_{k,s}$ since the attacker cannot corrupt too many of these labels. 
Second, our local decoder will test to make sure that the last node $k$ is {\em at least} $\alpha_2$-good. If this node is not $\alpha_2$-good then we must have found a red node and our {\em weak} decoder may output $\bot$; otherwise the last node serves as an {\em anchor point}. 
In particular, since label $\ell_{k,s} = \ell_{k,s}'$ in this case collision resistance now implies that for {\em any}  $\alpha_2$-good node $v$  then we must have $x_v' = x_v$ and $\ell_{v,s} = \ell_{v,s}'$ since $v$ {\em must} be connected to the node $k$ by a green path.

\paragraph{Strong \CRLCCs} The reason that the previous construction fails to yield a high rate Strong \CRLCCs is that it is possible for an attacker to {\em change} every node to a red-node by tampering with at most $\O{k\cdot \ell(\sec)/ \log k}$ bits. In particular, since the $\delta$-local expander $G$ has outdegree $\O{\log k}$ an attacker who tampers with just $\O{\ell(\sec)}$ bits could tamper with one of the labels so that $x_v' \neq x_v$. 
Now for {\em every} $w \in [k]$ s.t. $G$ contains the directed edge $(v,w)$ the node $w$ will be red and there are up to $\O{\log k}= \O{\log n}$ such nodes $w$. 

We address this issue by first applying a degree reduction gadget to our family $\{G_t\}_{t=1}^\infty$ of $\delta$-local expanders, where $t=\O{\frac{k}{\ell(\sec) \cdot \log k}}$, to obtain a new family of DAGs as follows: 
First, we replace every node $u \in [t]$ from $G_t$ with a chain $u_1,\ldots,u_m$ of $m=\O{\log t}$ nodes --- we will refer to node $u \in [t]$ as the meta-node for this group. 
The result is a new graph  $G$ with $k= m\cdot t$ nodes.  
Each of the nodes $u_1,\ldots,u_{m-1}$ will have constant indegree and constant outdegree. However, for $i<m$ we include the directed edges $(u_i,u_{i+1})$ and $(u_i, u_m)$ --- the node $u_m$ will have $\indeg(u_m) \in \Theta(m)$. Furthermore, if we have an edge $(u,v)$ in $G_t$ then we will add an edge of the $(u_m,v_j)$ for some $j < m$ --- this will be done in such a way that maintains $\indeg(v_j) \leq 2$ for each $j < m$. 
Therefore, the node $u_m$ will have $\outdeg(u_m) \in \Theta(\log t)$. 
We now note that if the label $\ell_{u_m,s}' = \ell_{u_m,s}$ and the node $u_m$ is green then, by collision resistance, it must be the case that $\ell_{u_i,s}' = \ell_{u_i,s}$ and $x_{u_i}' = x_{u_i}$ for {\em every} $u_i$ with $i \leq m$. 

With this observation in mind we tweak our last construction by grouping the values $\ell_{v_i,s}$ and $x_{v_i}$ into groups of size $m \in \O{\log t}$ before applying the error correcting code. 
In particular, given our input word $x$ divided into $k = m\cdot t$ distinct $\ell(\sec)$-bit strings $x_{v_i}$ for $i \leq m$ and $v \leq t$, our final output will consist of $c=\left(c_1,\ldots,c_{3t}\right)$ where $c_i = \ECC(x_{(i-1)m+1},\ldots,x_{im})$ for $i \leq t$ and $c_{v+t} = \ECC\left( \ell_{v_1, s} \circ \ldots \circ \ell_{v_m, s}\right)$ and $c_{2t+1} = \ldots = c_{3t} = c_{2t}$ consists of $t$ copies of $c_{2t}$. The final codeword $c$ will have length $n \in \O{tm \cdot \ell(\sec)} = \O{k \cdot \ell(\sec)} $ bits. 
By grouping blocks together we get the property that an attacker that wants to alter an individual label $\ell_{v_i, s}$ {\em must} pay to flip {\em at least} $\Omega\left( \ell(\sec) \cdot  m \right)$ bits of $c_{v+t}$. 
Thus, the attacker is {\em significantly} restricted in the way he can tamper with the labels e.g., the attacker cannot tamper with one label in every group. 

We say that the {\em meta-node} $u \in [t]=V(G_t)$ containing nodes $u_1, \ldots, u_m \in V(G)$ is  green if (1) node $u_m$ is green, and (2) at least $2m/3$ of the nodes $u_1,\ldots,u_m$ are green. 
We say that an edge $(u,v) \in E(G_t)$ in the meta-graph is red if the corresponding edge $(u_m,v_j) \in E(G)$ is incident to a red node; otherwise we say that the edge $(u,v)$ is green. We can then define the green meta-graph $G_g$ by dropping all red edges from $G_t$. 
We remark that, even if we flip a (sufficiently small) constant fraction of the bits in $c$ we can show that {\em most} of the meta-nodes must be green --- in fact, most of these meta-nodes must also be $\alpha_2$-good with respect to the set $R$ of red-nodes.  
In particular, there are only two ways to turn a meta-node $u$ red: (1) by paying  to flip $\O{\ell(\sec) \cdot m}$ bits in $c_u$ or $c_{u+t}$, and (2) by corrupting at least $m/3$ other meta-nodes $w$ such that the directed edge $(w,u)$ is in $G_t$. 

Furthermore, we can introduce the (unknown) set $T$ of tampered meta-nodes where for $i \leq 3t$ we say that $\tilde{c}_i$ is tampered if $\ECC\left(\ECCD(\tilde{c}_i)\right) \neq c_i$. In this case the hamming distance between $c_i$ and $\tilde{c}_i$ must be at least $\O{\ell(\sec) \cdot m}$ so the number $|T|$ of tampered meta-nodes cannot be too large. While the set $T$ is potentially unknown we can still show that {\em most} meta-nodes are $\alpha$-good with respect to the set $R \cup T$. 
In particular, there must exist some meta-node $\frac{3t}{4} \leq v \leq t$ s.t. $v$ is $\alpha$-good with respect to the set $R \cup T$. By collision resistance, it follows that for any meta-node $u < v$ which is connected to $v$ via a green path in $G_g$ we have $u \notin T$ e.g., $u$ is correct; otherwise we would have found a hash collision. 

While $G_t$ is a $\delta$-local expander the DAG $G_g$ may not be. 
However, we can show that, for suitable constants $\alpha$ and $\delta$, if a node $u$ is $\alpha$-good with respect to the set $R \cup T$ then $G_g$ has $2\delta$-local expansion around $u$. Furthermore, if $u < v$ has  $4\delta$-local expansion in $G_g$ then we can prove that $G_g$ contains a path to the meta-node $u$. 
Our local decoder will test whether or not $G_g$ has local expansion by strategically sampling edges to see if they are green or red. 
Our decoder will output $\bot$ with high probability if $u$ does {\em not} have $4\delta$-local expansion in $G_g$ which is the desired behavior since we are not sure if there is a green path connecting $u$ to $v$ in this case. 
Whenever $u$ does have $2\delta$-local expansion in $G_g$ the decoder will output $\bot$ with negligible probability which is again the desired behavior since there is a green path from $u$ to $v$ in this case (and, hence, $u \not T$ has not been tampered).  

\begin{remark}
Recall that in \thmref{thm:rlcc} and \thmref{strongRLCC} we construct a $({\polylog n}, \tau, \rho, \gamma,\mu(\cdot))$-Weak \CRLCC and a $({\polylog n}, \tau, \rho(\tau), \gamma(\cdot),\mu(\cdot))$-Strong \CRLCC where $\mu(\sec)$ represents the probability that the adversarial channel outputs a codeword $\tilde{c}$ and challenge index $i \in [k]$ that $\gamma(\sec)$-fools the local decoder. 
Intuitively, we can think of $\mu(\sec)$ as the probability that the codeword generated by the channel yields a hash collision and we can think of $\gamma(\sec)$ as the probability that we mis-identify a node (or meta-node) as being $\alpha$-good. 
If there are no hash collisions then the {\em only} way that the local decoding algorithm can be fooled is when it mis-identifies a node (or meta-node) as being $\alpha$-good. We show that {\em both} $\mu(\cdot)$ and $\gamma(\cdot)$ are negligible functions.
\end{remark}

%% file: lcc1.tex

\section{Construction of Weak-\CRLCC}
\seclab{sec:lcc}
In this section we overview the construction of a constant rate weak-\CRLCC scheme. 
In order to show the existence of a \CRLCC scheme, we need  to construct the \PPT algorithms $\Gen$, $\Enc$ and $\Dec$. 
Our construction will use a \CRHF $\Pi = (\GenH, H)$. 
In particular, $\Gen\left(1^\sec\right)$ simply runs $\GenH\left( 1^\sec\right)$ to output the public seed $s$. 

The $\Enc$ algorithm uses the \CRHF to create labels for the vertices of a $\delta$-local expander. 
We now define the recursive labeling process for the vertices of any given DAG $G$ using $H$. 
\begin{definition}[Labeling]
\deflab{def:lab}
Let $\Sigma = \{0,1\}^{\ell(\sec)}$. Given a directed acyclic graph $G=( [k],E)$, a seed $s \leftarrow \GenH\left(1^\sec\right)$ for a collision resistant hash function $H:\{0,1\}^* \rightarrow \{0,1\}^{\ell(\sec)}$,  and a message $x = (x_1, x_2, \ldots,  x_k) \in \Sigma^k$  
we define the labeling of graph $G$ with $x$, $\lab_{G,s}:\Sigma^k \rightarrow\Sigma^{k}$  as
$\lab_{G, s}(x) = (\ell_{1,s}, \ell_{2,s}, \ldots, \ell_{k,s})$, where 
\[\ell_{v,s}=
\begin{cases}
H(s, x_v),&\indeg(v)=0 \\
H(s, x_v\circ\ell_{v_1,s} \circ \ldots \circ \ell_{v_d,s}),& 0< \indeg(v)=d,
\end{cases}\]
where $v_1,\ldots,v_d$ are the parents of vertex $v$ in $G$, according to some predetermined topological order. 
We omit the subscripts $G,s$ when the dependency on the graph $G$ and public hash seed $H$ is clear from context.
\end{definition}

\subsection{$\Enc$ Algorithm }\seclab{sec:lccencode}
In this section we describe the $\Enc$ algorithm which takes as input a seed $s \leftarrow \Gen\left(1^\sec \right)$ and a message $x \in \{0,1\}^k$, and returns a codeword $c \in \{0,1\}^n$. 
For a security parameter $\sec$, let $H: \{0,1\}^* \rightarrow \{0,1\}^{\ell(\sec)}$ be a \CRHF (see \defref{def:crhf}). 
We will assume that $\Enc$ has access to $H$. 
In this section we also let $\ECC:\{0,1\}^{\ell(\sec)} \rightarrow \{0,1\}^{4 \cdot \ell(\sec)}$ be a standard binary code with message length $\ell(\sec)$ and block length $4 \cdot \ell(\sec)$ (such as Justesen code \cite{Justesen72}) and let $\ECCD:\{0,1\}^{4 \cdot \ell(\sec)} \rightarrow \{0,1\}^{\ell(\sec)}$ be the decoder that efficiently decodes from $\Delta_J$ fraction of errors.  

\begin{mdframed}
Let $s \leftarrow \Gen(1^\sec)$.  Let $k' = k/\ell(\sec)$ and for $\delta = 1/100$, let $G = ([k'], E)$ be a $\delta$-local expander graph with indegree $\O{\log k'}$.\\
\textbf{\Enc}$(s, x)$: \\
\underline{Input:}  $x = (x_1 \circ \ldots \circ x_{k'}) \in \{0,1\}^{k}$ where each $x_i \in \{0,1\}^{\ell(\sec)}$ and $k =k' \cdot \ell(\sec)$. \\
\underline{Output:} $c = (c_1 \circ \ldots \circ c_{3k'}) \in \{0,1\}^n$  where each $c_i \in \{0,1\}^{4\cdot \ell(\sec)}$ and $n = 4 \cdot 3k' \cdot \ell(\sec)$
\begin{itemize}
\item 
Let $\lab_{G,s}(x)= \left(\ell_{1,s} , \ldots, \ell_{k', s} \right) $ be the labeling of the graph $G$ with input $x$ and seed $s$.
\item  
The codeword $c$ consists of the block encoding of message $x$, followed by the encoded labeling of the graph using $\ECC$, followed by $k'$ copies of encoded $\ell_{k'}$. In particular, for each $1 \leq i \leq k'$ we have $c_i = \ECC(x_i)$, $c_{i+k'} = \ECC(\ell_{i,s})$ and $c_{i+2k'} = \ECC(\ell_{k',s}) = c_{2k'}$. 
\end{itemize}

\end{mdframed}
 
\begin{center}
\begin{figure*}[htb]
\centering
\begin{tikzpicture}[scale=1]
\draw[decorate,decoration={brace}](-3,-0.2) -- (-1,-0.2);
\draw[decorate,decoration={brace,mirror}](-3,-0.8) -- (-1,-0.8);
\node at (-2,-1.2){$4k'\ell(\sec)$ bits};
\node at (-2,0.6){Block encoding of};
\node at (-2,0.25){original message $x$};
\node at (-3,-0.5){\tiny{$\ECC({x_{1}})$}};
\node at (-2,-0.5){$\ldots$};
\node at (-1.2,-0.5){\tiny{$\ECC({x_{k'}})$}};
\node at (10,-0.5){\tiny{$\ECC(\ell_{k',s})$}};
\draw (0,0) circle (0.2);
\draw[->, black] (0.3,0) -- (0.7,0);
\node at (0.2,-0.5){\tiny{$\ECC(\ell_{1,s})$}};
\draw (1,0) circle (0.2);
\draw[->, black] (1.3,0) -- (1.7,0);
\node at (1,-0.5){\tiny{$\ldots$}};
\draw (2,0) circle (0.2);
\draw[->, black] (2.3,0) -- (2.7,0);
\node at (2,-0.5){\tiny{$\ldots$}};
\draw (3,0) circle (0.2);
\draw[->, black] (3.3,0) -- (3.7,0);
\node at (3,-0.5){\tiny{$\ldots$}};
\draw (4,0) circle (0.2);
\draw[->, black] (4.3,0) -- (4.7,0);
\node at (4,-0.5){\tiny{$\ECC(\ell_{5,s})$}};
\node at (5,0){$\ldots$};
\draw[->, black] (5.3,0) -- (5.7,0);
\filldraw[shading=radial,inner color=white, outer color=black!75, opacity=0.8] (6,0) circle (0.2);
\draw (6,0) circle (0.2);
\node at (5.7,-0.5){\tiny{$\ECC(\ell_{k',s})$}};
\filldraw[shading=radial,inner color=white, outer color=black!75, opacity=0.8] (7,0) circle (0.2);
\draw (7,0) circle (0.2);
\node at (7,-0.5){\tiny{$\ECC(\ell_{k',s})$}};
\draw[decorate,decoration={brace,mirror}](0,-0.8) -- (6,-0.8);
\node at (3.5,-1.2){$4k'\ell(\sec)$ bits};
\filldraw[shading=radial,inner color=white, outer color=black!75, opacity=0.8] (8,0) circle (0.2);
\draw (8,0) circle (0.2);
\node at (8.2,-0.5){\tiny{$\ECC(\ell_{k',s})$}};
\node at (9,0){$\ldots$};
\filldraw[shading=radial,inner color=white, outer color=black!75, opacity=0.8] (10,0) circle (0.2);
\draw (10,0) circle (0.2);
\node at (10,-0.5){\tiny{$\ECC(\ell_{k',s})$}};
\draw[decorate,decoration={brace}](7,0.4) -- (10,0.4);
\node at (8.5,0.7){$k'$ copies};
\draw[decorate,decoration={brace,mirror}](7,-0.8) -- (10,-0.8);
\node at (8.5,-1.2){$4k'\ell(\sec)$ bits};
\draw [->] (0,0.3) to [out=30,in=150] (2,0.3);
\draw [->] (0,0.3) to [out=45,in=135] (3,0.3);
\draw [->] (1,0.3) to [out=45,in=135] (4,0.3);
\draw [->] (3,0.3) to [out=45,in=135] (6,0.3);
\end{tikzpicture}
\caption{An example of an encoding with an underlying graph with $k' = k/\ell(\sec)$ nodes}
\end{figure*}
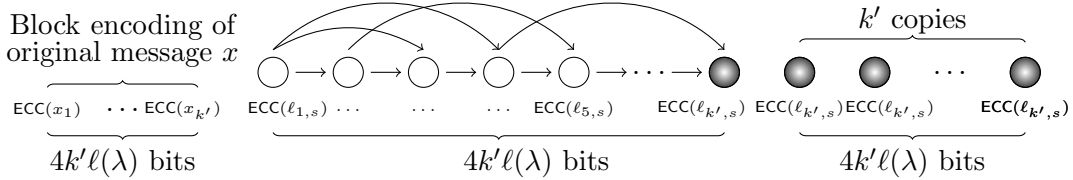
\end{center}
The parameter $\delta=1/100$ is an arbitrary simply chosen to satisfy \lemref{lem:cgood:connected}. 
Also, note that the codeword length is $n= 4 \cdot 3k' \cdot \ell(\sec)$ over the binary alphabet $\{0,1\}$, where $\ell(\sec)$ is the length of the output of the \CRHF $H$ and the original message had length $k = k' \cdot \ell(\sec)$. 
Therefore, the rate $\frac{k}{n}$ of \CRLCC scheme obtained from $\Enc$ is $\Theta(1)$.

%% file: lcc-dec.tex

\subsection{$\Dec$ Algorithm}
\seclab{sec:lccdecode}

In this section, we detail a randomized algorithm $\Dec: \{ 0,1\}^{n} \times [n] \rightarrow  \{\bot \} \cup \{0,1\}$ for $\Enc: \{ 0,1\}^{k} \rightarrow \{ 0,1\}^{n}$ described in \secref{sec:lccencode}. 
We assume $\Dec$ has access to the same public seed $s$ used by $\Enc$ as well as the $\delta$-local expander used by $\Enc$. 
We focus on the key ideas first and then provide the formal description of $\Dec$. 
The notion of a green node is central to $\Dec$. 
 
\begin{definition}[Green/red node] Given a (possibly corrupted) codeword $w=(w_1 \circ \ldots \circ w_{3k'}) \in (\{ 0,1\}^{4 \cdot \ell(\sec)})^{3k'}$ we define $x_i' = \ECCD(w_i) \in \{0,1\}^{\ell(\sec)}$ for $i \leq k'$ and $\ell_{v,s}' = \ECCD(w_{k'+v})$ for $v \leq k'$ and $\ell_{k',s,i}' = \ECCD(w_{2k'+i})$ for $i \leq k'$. 
 We say that a node $v \in [k']$ with parents $v_1,\ldots, v_d$ is {\em green} if the label $\ell_{v,s}'$ is consistent with the hash of its parent labels i.e., $\ell_{v,s}' = H(s, x_v' \circ \ell_{v_1, s}' \circ \ell_{v_2,s}' \circ \ldots \circ \ell_{v_d, s}')$. A node that is not green is a {\em red} node.
\end{definition}

We say that a node $v$ is {\em correct} if $\ell_{v,s}' = \ell_{v,s}$.
\noindent
\lemref{lem:cgood} highlights one of the key reasons why green nodes are significant in our construction. 

\begin{lemma}{\lemlab{lem:cgood}}
Suppose that node $v$ is {\em green} and {\em correct} (i.e., $\ell_{v,s}' = \ell_{v,s}$) and suppose that there is a directed path from node $u$ to node $v$ consisting {\em entirely} of green nodes. 
Then either the node $u$ is also correct (i.e., $\ell_{u,s}' = \ell_{u,s}$) or the \PPT adversary has produced a hash collision.
\end{lemma}
\begin{proof}
If node $v$ is green, then by definition $\ell_{v,s}' = H(s,x_v' \circ \ell_{v_1,s}' \circ \ldots \circ \ell_{v_d,s}')$ where $v_1,\ldots,v_d$ are the parents of node $v$. Moreover, if $\ell_{v,s}'$ is unaltered, then $\ell_{v,s}' = \ell_{v,s} $, \emph{i.e.}, 
\[H(s,x_v \circ \ell_{v_1,s} \circ \ell_{v_2,s} \circ \ldots \circ \ell_{v_d,s}) = H(s,x_v' \circ \ell_{v_1,s}' \circ \ldots \circ \ell_{v_d,s}').\]
So, either the adversary has successfully found a hash collision, or  $x_v = x_v'$ and $\ell_{v_j,s} = \ell_{v_j,s}'$ for each $j \in [d]$. 
Let $v_j$ be the parent node of $v$ that lies on the green path from $u$ to $v$. Assuming we have no hash collision, we can conclude that the node $v_j$ is {\em both} correct and green.  Extending the same argument iteratively along the green path from $u$ to $v_j$, we get that either node $u$ must be correct, or the adversary would have found a hash collision. 
\end{proof}

Our local decoder will output $\bot$ if it ever detects any red nodes (Note that when the codeword is not corrupted we will have $0$ red nodes). 

We now consider two cases.
{\bf Case 1:} The input index $i \ge 8k-4\ell(\sec)$. 
This corresponds to a bit query within the last $k'$ blocks. 
From the construction of the code, we know that the last $k'$ blocks of the encoding are the same, i.e. $c_{2k' +j} = c_{2k'}$ for all $1 \le j \le  k'$. 
In this case, $\Dec^{w}$ simply queries some blocks in $[2k', 3k']$, decodes them to the nearest codeword and returns the corresponding bit of the majority codeword. Since the adversary cannot corrupt many blocks beyond the decoding radius, the majority will return the correct codeword with high probability. See \appref{app:i-ge-8k} for the formal proof.

{\bf Case 2:} If the input index $i< 8k - 4\ell(\sec)$, $\Dec^{w}$ uses the properties of the $\delta$-local expander graph to detect whether the corresponding label has been tampered or not. 
If the label is tampered, then the decoder returns a $\bot$ else it returns the $i$-th of the received vector. 
Let $i'$ denote the index of the node in $G$ associated with the queried bit $i$ \emph{i.e.}, $i' = \left\lceil \frac{i}{4 \cdot \ell(\sec)}\right \rceil  \mod{k'}$. 
We check to make sure that (1) the node $k'$ is {\em at least} $\alpha$-good (w.r.t. the set of red nodes), and (2) the node $i'$ is also $\alpha$-good (for some $\alpha < 1-2\delta$). 
If either check fails then we will output $\bot$. 
Otherwise, by \lemref{lem:cgood} the node $i'$ must be correct (with high probability) since (1) there is a green path connecting any two $\alpha$-good nodes in a $\delta$-local expander, and (2) the node $k'$ must be correct (with high probability). 

The core part of the $\Dec^{w}$ then is the probabilistic procedure to verify that  a node $v$ is $\alpha$-good. 
First, it is clear that there is a deterministic procedure $\Green^{w}(v)$ that checks whether a given node $v \in [k']$ is green or not using $\O{\log n}$ coordinate queries to $w$ (see \lemref{lem:green} for the formal statement). 
Now in order to verify if a node $v$ is $\alpha$-good, we need to estimate the fraction of green nodes in any neighborhood of it. 
This is achieved by designing a tester that \emph{accepts} $v$ if it is $\alpha/4$-good with respect to the set $S$ of red nodes and rejects with high probability if $v$ is not $\alpha$-good (formalized in \lemref{lem:cgood:check}).

The key observation behind  \lemref{lem:cgood:check} is that it suffices to check that $\left|[v,v+2^j-1] \cap S \right| \leq \alpha \cdot  2^j/4$ for each $j \leq \log k'$ and that $\left|[v-2^j+1,v] \cap S \right| \leq \alpha \cdot  2^j/4$ for each $j \leq \log k'$. 
For each $j$ we can sample  $r=\O{\polylog k'}$ random nodes (with replacement) in each of the intervals $[v-2^j+1,v] $ and $[v,v+2^j-1] $ and use the subroutine $\Green$ to count the number of red (resp. green) nodes in each interval. 
If for every $j \leq \log k'$ the number of red nodes is both intervals {\em at most} $\alpha \cdot 2^j/2$ we accept. 
Otherwise, we reject.  

%% file: lcc-final.tex
\section{Strong-\CRLCC}\seclab{sec:strong}
In this section, we give an improved construction that locally corrects a majority of input coordinates. 

{\bf The Barrier:} The challenge in correcting a constant fraction of the nodes in the previous construction is that an adversary that is permitted  to change  $\O{\frac{n}{\log n}}$ symbols can turn all nodes red in a graph that has outdegree $\O{\log n}$. 

{\bf Key Idea:} We overcome the above mentioned barrier using three ideas. 
First, we run the original $\delta$-local expander $G_t$ through an indegree reduction gadget to obtain a new graph $G$ with $k'=t\cdot m$ nodes in which {\em almost every} node has at most $2$ incoming/outgoing edges and each node $u$ in $G_t$ (called the \emph{meta-node}) corresponds to $m$ nodes $u_1,\ldots,u_m$ in $G$. 
During the encoding process we group the labels. 

Second, we group the labels $\ell_{u,s} = \left( \ell_{u_1,s},\ldots, \ell_{u_m,s}\right)$ and data values $x_u = \left( x_{u_1},\ldots, x_{u_m}\right)$ and include the encodings $c_u= \ECC( x_u)$ and $c_{u+t} = \ECC(\ell_{u,s})$ in our final codeword. 
Notice that if we receive a the codeword $c' \neq c$ we can first preprocess each block $c_i'$ to obtain $\tilde{c}_i = \ECC\left(\ECCD(c_i')\right)$. 
Thus, if $\ECC\left(\ECCD(c_i')\right) = c_i$ we will say that the block $i$ is untampered --- even if $c_i' \neq c_i$. 
We say that a meta-node $u$ is tampered if either of the blocks $c_u$ or $c_{u+t}$ were tampered and we use $T$ to denote the set of tampered meta-nodes. 
We can show that most meta-nodes are not contained in the set $T$ since the attacker must alter a constant fraction of the bits in $c_u$ or $c_{u+t}$.   
In particular, the attacker has to flip {\em at least} $\Omega\left( m \cdot \ell(\sec)\right)$ bits to tamper {\em any} of the meta-nodes. 
In this way we can ensure that an attacker must tamper with {\em at least} $\Omega\left( t\cdot m \cdot \ell(\sec)\right)$ bits to make $\Omega(t)$ meta-nodes red.

Third, we introduce the notion of the {\em green} subgraph $G_g$ of the meta-graph $G_t$ (see \defref{def:greengraph}). 
In particular, $G_g$ is obtained from $G_t$ by discarding all red-edges (an edge $(u,v) \in E(G_t)$ is red if the corresponding edge $(u_m,v_j) \in E(G)$ in $G$ is incident to a red-node).  
As before a node $v_i$ in $G$ is green if its hash value is locally consistent e.g., $\ell_{v_i,s}' = H(s \circ x_{v_i}' \circ \ell_{v_i^1}' \circ \ldots \circ \ell_{v_i^d}' )$ where $v_i^1,\ldots, v_i^d$ are the parents of node $v_i$ in $G$ and for each meta-node $u$ in $G_t$ we define $\ell_u'=\left(\ell_{u_1,s}',\ldots, \ell_{u_m,s}' \right)=\ECCD(c_{u+t}')$ and $x_u'=\left(x_{u_1}',\ldots, x_{u_m}' \right) =\ECCD(c_u')$. 
We say that a meta-node $u$ in $G_t$ is green if at least $\frac{2}{3}$ of the corresponding nodes $\{u_1,\ldots,u_m\}$ in $G$ are green {\em and} the last  node $u_m$ is also green (\defref{def:meta:green}). 
We show that the same key properties we used for the construction of Weak-\CRLCC still hold with respect to the meta-nodes i.e., if there is a green path connecting meta-node $u$ to meta-node $v$ and the meta-node $v$ is untampered (i.e., $c_v = \ECC\left(\ECCD(c_v')\right)$ and $c_{v+t} = \ECC\left(\ECCD(c_{v+t}')\right)$) then the meta-node $u$ is also untampered (i.e., $c_u = \ECC\left(\ECCD(c_u')\right)$ and $c_{u+t} = \ECC\left(\ECCD(c_{u+t}')\right)$) ; otherwise we would have found a hash collision (\lemref{lem:connected:correct}). 
If we let $R$ denote the subset of red meta-nodes and if we let $T$ denote the subset of (unknown) tampered meta-nodes we have $G_g$ contains the graph $G_t-(R\cup T)$. 
One can prove that most meta-nodes are $\alpha$-good with respect to the set $R \cup T$ and that if a meta-node $u$ is $\alpha$-good with respect to the set $R \cup T$ then the green graph $G_g$ also has $2\delta$-local expansion around the meta-node $u$ (\lemref{lem:alpha:expansion}). 
Finally, if $G_g$ has $\delta'$-local expansion around both of the meta-nodes $u$ and $v$ with $\delta' \leq 4\delta$  then there is a green path connecting $u$ and $v$ (\lemref{lem:greenconnected}).

The remaining challenge is that the set $T$ is unknown to the local decoder\footnote{By contrast, it is easy to check if a node (resp. meta-node) is red/green by checking if the hash value(s) are locally consistent.}. 
Thus, we cannot directly test if a node is $\alpha$-good with respect to $R \cup T$. 
What we can do is develop a test that accepts all meta-nodes that are $\alpha$-good with respect to $R \cup T$ with high probability and rejects all meta-nodes $u$ that do not have $4\delta$-local expansion in $G_g$ by sampling the number of red/green edges in between $A_i=[u,u+2^i-1]$  and $B_i = [u+2^i,u+2^{i+1}-1]$ for $i \geq 1$. 
In particular, we restrict our attention to edges in the $\delta$-expander graph $H_{i,\delta}$ connecting $A_i$ and $B_i$ in $G_t$, where $H_{i,\delta}$ has maximum indegree $d_\delta$ for some constant $d_\delta$ which depends on $\delta$. 
If $G_g$ is not a $4\delta$-local expansion around $u$ then for some interval $i$ a large number of red edges between $A_i$ and $B_i$ {\em must} have been removed --- at least $c_2 \delta 2^{i}$ red edges out of at most $2^i\times d_\delta$ edges in $H_{i,\delta}$. 
By contrast, if $u$ is $\alpha$-good with respect to $R \cup T$ then we can show that at most $c_1 \alpha 2^{i} d_\delta$ edges are deleted. For a suitably small value of $\alpha < \frac{\delta}{d_\delta}$ we can ensure that $c_1 \alpha 2^{i} d_\delta < 4 c_2 \delta 2^i$. 
Thus, our tester can simply sample random edges in $H_{i,\delta}$ and accept if and only the observed fraction $f_{r,i}$ of red edges is at most, say, $2 c_1 \alpha 2^i d_\delta$ (\lemref{lem:green:expander}). 

The full description of the decoder is in \appref{app:strong:dec}.
 
\subsubsection*{Graph Degree Reduction}
We now describe our procedure \ReduceDegree that takes as input a $t$ node $\delta$-local expander DAG $G_0$ with indegree $m= O(\log t)$ and outputs a new graph $G$ with $mt$ nodes. 
$G$ has the property that {\em most} nodes have indegree two.
\begin{mdframed}
\textit{\ReduceDegree}: (Let $G_0$ be a $\delta$-local expander.) \\
\underline{Input: } Graph $G_0$ with $t$ vertices and $m = \max\{\indeg(G_0), \outdeg(G_0)\}+1$ \\
\underline{Output: } Graph $G$ with $m\cdot t$ vertices.
\begin{itemize}
\item For each node $u$ in $G_0$ we add $m$ nodes $u_1,\ldots,u_m$ to $G$. Along with each of the directed edges $(u_i,u_{i+1})$ and $(u_i, u_m)$ for $i < m$.
\item For each directed edge $(u,v)$ in $G_0$ we connect the final node $u_m$ to the first node $v_j$ that currently has indegree {\em at most} $1$.  
\end{itemize}
\end{mdframed}
We call $G_0$ the {\em meta-graph} of $G$ and we refer to nodes in $G_0$ as {\em meta-nodes}. 
In particular, the meta-node $u \in V(G_0)$ corresponds to the {\em simple} nodes $u_1,\ldots,u_m \in V(G)$. 
Notice that for two meta-nodes $u < v \in V(G_0)$ there is an edge $(u,v) \in E(G_0)$ in the meta-graph if and only if $G$ has an edge of the form $(u_m,v_i)$ for some $i \leq m$. 
While our encoding algorithm will use the graph $G$ to generate the labeling, it will be helpful to reason about meta-nodes in our analysis since the meta-graph $G_0$ is a $\delta$-local expander. 

\begin{figure*}[htb]
\centering
\begin{tikzpicture}[scale=0.9]
\node at (0, 0){$\cdots$};
\draw[->, black] (1,0) -- (2,0);
\draw (3,0) circle (0.4);
\node at (3,0){$u$};
\draw[->, black] (4,0) -- (5,0);
\node at (6, 0){$\cdots$};
\draw[->, black] (7,0) -- (8,0);
\draw (9,0) circle (0.4);
\node at (9,0){\tiny{$v-1$}};
\draw[->, black] (10,0) -- (11,0);
\draw (12,0) circle (0.4);
\node at (12,0){$v$};
\draw [->] (3.3,0.4) to [out=20,in=160] (11.7,0.4);

\draw (2.5,-0.3) -- (2, -1);
\draw (3.5,-0.3) -- (4, -1);
\draw (3,-2) ellipse (1.4 and 1);
\draw (2.7,-2) circle (0.2);
\node at (2.7,-2){\tiny{$u_2$}};
\draw (3.8,-2) circle (0.2);
\node at (3.8,-2){\tiny{$u_m$}};
\draw (2,-2) circle (0.2);
\draw [->] (2.2,-2) -- (2.5,-2);
\node at (2,-2){\tiny{$u_1$}};
\node at (3.3, -2){$\cdots$};
\draw [->] (2,-1.7) to [out=20,in=160] (3.7,-1.7);
\draw [->] (2.7,-2.3) to [out=-20,in=-160] (3.7,-2.3);

\draw (8.5,-0.3) -- (8, -1);
\draw (9.5,-0.3) -- (10, -1);
\draw (9,-2) ellipse (1.4 and 1);
\draw (8.7,-2) circle (0.2);
\draw (9.8,-2) circle (0.2);
\draw [->] (8.2,-2) -- (8.5,-2);
\draw (8,-2) circle (0.2);
\node at (9.3, -2){$\cdots$};
\draw [->] (8,-1.7) to [out=20,in=160] (9.7,-1.7);
\draw [->] (8.7,-2.3) to [out=-20,in=-160] (9.7,-2.3);

\draw (11.5,-0.3) -- (11, -1);
\draw (12.5,-0.3) -- (13, -1);
\draw (12,-2) ellipse (1.4 and 1);
\draw (11.7,-2) circle (0.2);
\node at (11.7,-2){\tiny{$v_2$}};
\draw (12.8,-2) circle (0.2);
\node at (12.8,-2){\tiny{$v_m$}};
\draw (11,-2) circle (0.2);
\draw [->] (11.2,-2) -- (11.5,-2);
\node at (11,-2){\tiny{$v_1$}};
\node at (12.3, -2){$\cdots$};
\draw [->] (11,-1.7) to [out=20,in=160] (12.7,-1.7);
\draw [->] (11.7,-2.3) to [out=-20,in=-160] (12.7,-2.3);

\draw[->, black] (4.7,-2) -- (5.7,-2);
\node at (6, -2){$\cdots$};
\draw[->, black] (6.3,-2) -- (7.3,-2);

\draw [->] (10,-2) -- (10.8,-2);
\draw [->] (4,-2.3) to [out=-20,in=-160] (10.8,-2.3);

\draw (-0.5,-0.5) rectangle +(14.3,1.9);
\node at (0, 0){$G_0$};
\draw (-0.5,-3.3) rectangle +(14.3,2.5);
\node at (0, -2){$G$};
\end{tikzpicture}
\caption{An example of a degree reduction graph}
\end{figure*}
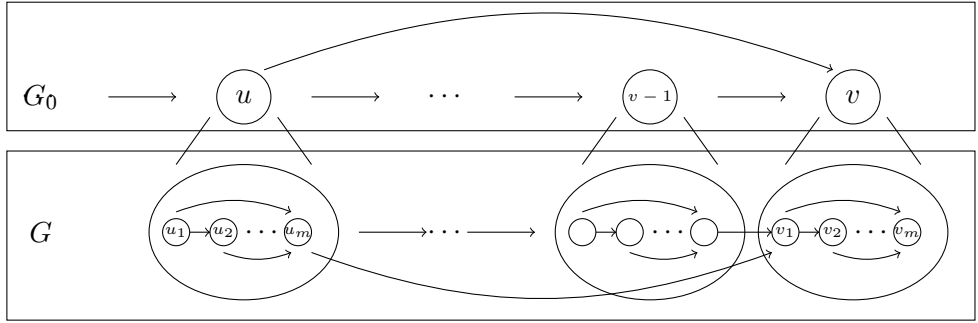

\subsection{$\Enc$ Algorithm }
\seclab{sec:maj:lccencode}
In this section we describe the $\Enc$ algorithm which takes as input a seed $s \leftarrow \Gen\left(1^\sec \right)$ and a message $x \in \{0,1\}^k$ and returns a codeword $c \in \{0,1\}^n$. 
For a security parameter $\sec$, let $H: \{0,1\}^* \rightarrow \{0,1\}^{\ell(\sec)}$ be a \CRHF (see \defref{def:crhf}). 
Recall that, $\ECC$ and $\ECCD$ denote the encoding and decoding algorithms of a good code of rate $R=R(\tau)$ 
that can decode efficiently from some constant $\Delta_J = \Delta_J(\tau)$ fraction of errors (Note: as $\tau \rightarrow 0$ we have $\Delta_J(\tau) \rightarrow 0$ and $R(\tau) \rightarrow 1$). 

Let $\beta= \beta(\tau) > 0$ be a parameter which can be tuned. 
For $t=\O{\frac{k}{\beta \ell(\sec) \cdot \log k}}$, let $G_0 = ([t], E)$ be a $\delta$-local expander graph on $t$ vertices and degree (indegree and outdegree) $m=\O{\log t}$ where $\delta=\delta(\tau)< \frac{1}{4}$ is a parameter than can be tuned as needed\footnote{As $\delta$ grows smaller the degree of $G_0$ increases. 
In particular, we have $\indeg(G_0) = d_\delta \log n$ where $d_\delta$ is some constant which depends on $\delta$ --- as $\delta$ decreases $d_\delta$ increases. We will also have $m = d_\delta \log n$  nodes in each metanode of $G$ so $m$ increases as $\delta$ decreases. 
This will increase the {\em locality} of our decoding algorithm, but only by a constant factor.}. 
Let $G:=\ReduceDegree(G_0)$ be the graph with $k' := \frac{k}{ \beta \ell(\sec)}$ nodes output by the degree reducing procedure applied to $G_0$. 
Crucially, the graph has the property that for $\alpha = \alpha(\tau) < 1-2\delta(\tau)$ and any red/green coloring of the nodes of the graph that any pair of $\alpha$-good meta-nodes with respect to the set $T$ of tampered nodes are connected by a {\em green} path.  

\begin{mdframed}
Let $s \leftarrow \Gen(1^\sec)$ \\
\textbf{\Enc}$(s, x)$: \\
\underline{Input:}  $x = (x_1\circ \ldots \circ x_{k'}) \in \left( \{0,1\}^{ \beta\ell(\sec)} \right)^{k'}$, where $k = k' \cdot \beta \ell(\sec)$
\\
\underline{Output:} $c = (c_1 \circ \ldots \circ c_{3t}) \in \{0,1\}^{n}$, where $t=k'/m$ and $n = \frac{tm \ell(\sec)}{R} \cdot (\beta+2)$.
\begin{itemize}
\item Let $\left( \ell_{1_1,s},\ldots,\ell_{1_m,s},\ldots,\ell_{t_1,s}\ldots,\ell_{t_m, s}\right) = \lab_{G, s}(x)$ be the labeling of the graph $G$ with input $x$ using the \CRHF, H (see \defref{def:lab}). 
\item Let $U_1 :=(\ell_{1_1,s} \circ \ell_{1_2,s} \circ \ldots \circ \ell_{1_m,s}),  \ldots, U_t :=(\ell_{t_1,s} \circ \ell_{t_2,s} \circ \ldots \circ \ell_{t_m,s})$.
\item Let $T_1:=(x_1 \circ x_2 \circ \ldots \circ x_m), \ldots, T_t:=(x_{(t-1)m+1} \circ x_{(t-1)m+2} \circ \ldots \circ x_{tm})$.
\item The codeword $c$ consists of the encoding of groups of message bits, followed by  encoding using $\ECC$ of groups of labels, followed by $t=\frac{k'}{m}$ copies of the last encoding. 
\begin{align*}
 c &= (\ECC(T_1) \circ \ECC(T_2) \circ  \ldots \circ \ECC(T_{t}) \circ \ECC(U_1) \circ \ECC(U_2) \circ \ldots \circ \ECC(U_{t}) \circ \\
&\underbrace{\ECC(U_{t}) \circ \ldots \circ \ECC(U_{t})}_{t \text{ times}}), \\
&\text{ where } \ECC(U_i)=\ECC( \ell_{i_1,s} \circ \ldots \circ \ell_{i_m,s} ) \text{ for all } i \in [t].
\end{align*}
\end{itemize}
\end{mdframed}

The codeword $c = \Enc(x)$ consists of 3 parts. 
The first $k/R$ bits correspond to the message symbols $x$ passed $m \cdot \beta \ell(\sec)$ bits at a time through $\ECC$ of rate $R$. 
The next $k/(\beta R)$ bits correspond to $t$ codewords $\ECC(U_j)$, and finally the last $k/(\beta R)$ bits correspond to the repetitions of the final \ECC codeword, $\ECC(U_{t})$.
 Therefore the length of any codeword produced by $\Enc$ is $n= \frac{k}{R} \cdot \frac{\beta+2}{\beta}$. The information rate of the constructed code is therefore $R \cdot \left(\frac{\beta}{\beta+2}\right)$.  Therefore appropriately choosing the parameters $\beta$ and $R$, we get the rate  of the \CRLCC approaching $1$.\\ 
 
{\noindent \bf Strong \CRLCC Decoder:} We can show that for any $w\in \{0,1\}^n$ with $\HAM(w,c)\le\frac{\Delta_J \cdot k}{4}$  the number of tampered meta-nodes $T \subset V(G_0)$ is very small. 
If let $t(\tau)$ denote the maximum number of meta-nodes $v$ that the attacker can either corrupt or ensure that $v$ is not $\alpha$-good with respect  to $T$ then it can be shown that $t(\tau)/t \rightarrow 0$ as $\tau \rightarrow 0$ i.e., almost all of the meta-nodes $u$ are $\alpha$-good with respect to the set $T$ of tampered meta-nodes and that any such node $u$ that has $2\delta$-local expansion in the green subgraph $G_g$. 
We also show that {\em any} meta-node $v < t - t(\tau)-1$ that has $\delta'$-local expansion with $\delta' \leq 4\delta$ must also be correct because at least one of the last $t(\tau)+1$ nodes must be $\alpha$-good and correct and $v$ will be connected to this node via an all green path since we have $4\delta$-local expansion around both nodes in the green subgraph $G_g$. 
While the set $T$ is potentially unknown it is possible to test which (with high probability) accepts all $\alpha$-good meta-nodes and (with high probability) rejects any meta-node $u$ with the property that the green subgraph $G_g$ does not have $4\delta$-local expansion around $u$. 
We can distinguish between the two cases by strategically sampling edges in each of the intervals $A_i=[u,u+2^i-1]$ and $B_i=[u+2^i,u+2^{i+1}-1]$. 
The original meta-graph $G_0$ contains a $\delta$-bipartite expander $H_i$ connecting $A_i$ and $B_i$. 
If $u$ does not have $4\delta$-local expansion in $G_g$ then for some $i$ a large fraction of the edges in $H_i$ must be red. 
By contrast, if $u$ is $\alpha$-good with respect to $T$ then we can show that the number of red edges in $H_i$ will be much smaller. 
Thus, we can ensure that we can decode at least $\rho(\tau) \geq \left(1-2\times t(\tau)/t\right)$ fraction of the bits in the original codeword so that $\rho(\tau) \rightarrow 1$ as $\tau \rightarrow 0$. 
With this observation it is relatively straightforward  to extend the ideas from the previous section to build a strong $\Dec$ with locality $\O{\polylog n}$. 
A full description of the strong \CRLCC decoder can be found in \appref{app:strong:dec}.

%% file: relwork.tex

\section{Expanded Discussion of Related Work}\label{sec:relwork}

\paragraph{Classical \LDCs/\LCCs and Relaxed Variants in the Hamming channel} 
The current best constructions for \LDCs/\LCCs can achieve any constant rate $R<1$, any constant relative distance (i.e., minimum Hamming distance between codewords)  $\delta<1-R$, and query complexity $\O{2^{\sqrt{\log n \log \log n}}}$ \cite{KoppartyMRS17}.  
In the constant query complexity regime, for $q\geq 3$, the best codes achieve blocklength  subexponential in the message length \cite{Yekhanin08,Efremenko12,DvirGY11}. For $q=2$, Kerenidis and deWolf \cite{KerenidisW04} show  
an exponential in $k$ lower bound for the blocklength of any \LDC. 

A notion similar to \RLDCs, called Locally Decode/Reject code, was studied by Moshkovitz and Raz \cite{MoshkovitzR10} in the context of building better \PCPs. The notion concerns decoding batches of coordinates jointly, and it allows as output a list of possible messages. 
In the context of simplifying the proofs in \cite{MoshkovitzR10}, a related notion of decodable \PCPs was studied in \cite{DinurH09}.

\paragraph{Non-Local Codes in Computationally Bounded Channels} In their initial works in  the computationally bounded channel, Lipton \cite{Lipton94} and Ding, Gopalan and Lipton \cite{GopalanLD04} obtain error-correcting codes uniquely decodable (in the classical, global setting) with error rates beyond what is possible in the adversarial  Hamming channel. 
Their model requires that the sender and receiver share a {\em secret}  random string, unknown to the channel. 
This is a strong cryptographic assumption, as the model is unsuitable to common settings such as message broadcast. Micali \etal \cite{MicaliPSW05} 
address this drawback by proposing public-key error-correcting codes against computationally bounded channels. 
Their model is based on the observation that if one starts with a code that is {\em list-decodable}, by encoding  messages using a secret key and digitally signing them, one can prevent a \PPT channel from producing valid signatures, while allowing the receiver to pick the unique message from the list with a valid signature. 
In follow-up work,  Guruswami and Smith \cite{GuruswamiS16} removes the public-key setup, and obtains optimal rate error-correcting codes for channels that can be described by simple circuits. 
Their channels are either ``oblivious'', namely the error is independent of the actual message being sent, or they are describable by polynomial size circuits. 
Their results are based on the idea that the sender can choose a permutation and a ``blinding factor'' that are then embedded into the codeword together with the message. 
The channel cannot distinguish the hidden information  since  it operates with low complexity, while the receiver can.

\paragraph{\LDCs in Computationally Bounded Channels} The notion of \LDCs over computationally bounded channels was introduced in \cite{OstrovskyPS07}, where the authors define private \LDCs. 
In these constructions the sender and the receiver share a private key. 
They obtain codes of constant rate and  $\O{\polylog \lambda}$ query complexity, where $\lambda$ is the security parameter.  
Hemenway and Ostrovsky \cite{HemenwayO08} build \LDCs in the public-key model, and obtain codes of constant rate, and $\O{\lambda^2}$ locality. 
Hemenway \etal \cite{HemenwayOSW11} improve on these results and obtain public-key \LDCs with constant rate, constant error rate and $\O{\lambda}$ locality, which work under a weaker cryptographic assumption than that of \cite{HemenwayO08}. 

\paragraph{Other Applications of Depth-Robust Graphs} Depth-robust graphs have found many applications in cryptography. 
These applications includes the construction of memory hard functions with {\em provably high} cumulative memory complexity~\cite{STOC:AlwSer15,C:AlwBlo16,EC:AlwBloPie17,CCS:AlwBloHar17,TCC:BloZho17}, sustained space complexity~\cite{EC:AlwBloPie18} and bandwidth hardness~\cite{EPRINT:BloRenZho18,TCC:RenDev17}. 
Depth robust graphs have also been used to construct {\em proofs of space}~\cite{C:DFKP15} and they were also used in the first construction of a {\em proof of sequential work}~\cite{ITCS:MahMorVad13} although Cohen and Pietrzak~\cite{EC:CohPie18} found an alternate construction of a Proofs-Of-Sequential Work protocol that does not involve depth-robust graphs. 
Finally, depth-robust graphs have also been used to derive cumulative space in the black-white pebble game \cite{ITCS:AdNV17,EC:AlwBloPie18} which is of interest in the study of proof complexity. 
To the best of our knowledge we are the first to apply depth-robust graphs (more specifically $\delta$-local expanders) in the area of coding theory to construct (relaxed) locally correctable codes. 

%% file: prelims.tex
\section{Additional Background: Justesen Codes, CRHFs and Local Expanders}
\applab{app:prelims}
We use the notation $[n]$ to represent the set $\{1,2,\ldots,n\}$. 
For any $x, y \in \Sigma^n$, let $\HAM(x)$ denote the Hamming weight of $x$, i.e. the number of non-zero coordinates of $x$. 
Let $\HAM(x, y) = \HAM(x-y)$ denote the Hamming distance between the vectors $x$ and $y$. 
For any vector $x \in \Sigma^n$, let $x[i]$ be the $i$\th coordinate of $x$. 
We also let $x\circ y$ denote the concatenation of $x$ with $y$. 

We denote a directed acyclic graph $G$ with $n$ vertices labelled in topological order by  $G=([n],E)$. 
A node $v\in V$ has indegree $\d = \indeg(v)$ if there exist $\d$ incoming edges $\d = |(V \times \{v\}) \cap E|$. 
Thus, we say that graph $G$ has indegree $\d=\indeg(G)$ if the maximum indegree of any node of $G$ is $\d$. 
A node with indegree $0$ is called a source node and a node with no outgoing edges is called a sink. 
A node $v\in V$ has outdegree $\d = \outdeg(v)$ if there exist $\d$ outgoing edges $\d = |(V \times \{v\}) \cap E|$. 
Thus, we say that graph $G$ has outdegree $\d=\outdeg(G)$ if the maximum outdegree of any node of $G$ is $\d$. 
Finally, we say that graph $G$ has degree $\d=\deg(G)$ if the maximum degree of any node of $G$ is $\d$, or equivalently $\max_{v\in V}\outdeg(v)+\indeg(v)=\d$. 
We use $\parents_G(v) = \{u \in V: (u,v) \in E\}$ to denote the parents of a node $v \in V$. 
We will often consider graphs obtained from other graphs by removing subsets of nodes. 
Therefore if $S \subset V$, then we denote by $G-S$ the DAG obtained from $G$ by removing nodes $S$ and incident edges. 

Let $\mathcal{C}$ be a $(n,k)$ code that maps any $k$-length \emph{message} over alphabet $\Sigma$ to a unique $n$-length \emph{codeword} over alphabet $\Sigma$. 
We say $n$ is the \emph{block length} of the code, and $k/n$ is the {\em information rate}. 
Let $\Enc: \Sigma^k \rightarrow \Sigma^n$ denote the encoding map of $\mathcal{C}$. 
We define the {\em minimum distance} of $\mathcal{C}$ to be the quantity $\min_{c_1, c_2 \in \mathcal{C}} \HAM(c_1, c_2)$ and the {\em  relative minimum distance} of $\mathcal{C}$ to be the quantity $\min_{c_1, c_2 \in \mathcal{C}}\frac{\HAM(c_1, c_2)}{n}$.   

We shall use $\Enc$ and $\Dec$ to refer to the encoding and decoding algorithms of our construction, and $\ECC$ and $\ECCD$ to refer to the encoding and decoding algorithms for a binary code $C_J$ with constant rate and relative distance. 
We will use Justesen codes  \cite{Justesen72} in what follows.

\begin{theorem}\cite{Justesen72}
For any $0< R < 1$, there exist binary linear codes of rate $R$, that are efficiently decodable from $\Delta_J(R)$ fraction of errors, where $\Delta_J(R)$ is a function that only depends on $R$. 
\end{theorem}

\subsection*{Collision Resistant Hash Functions (\CRHF)}
Our code constructions involve the use of collision-resistant hash functions. 
We use the following definitions from Katz and Lindell \cite{KatzL14}.
\begin{definition}
\cite[Definition 4.11]{KatzL14}
A \emph{hash function} with alphabet $\Sigma$ and blocklength $\ell(.)$ is a pair $\Pi= (\GenH,H)$ of probabilistic polynomial time algorithms  satisfying:
\begin{itemize}
\item
$\GenH$ is a probabilistic algorithm which takes as input a security parameter $1^\sec$ and outputs a public seed $s \in \{0,1\}^*$, where the security parameter $1^\sec$ is implicit in the string $s$.
\item $H$ is a deterministic algorithm that takes as input a seed $s$ and a string $\Sigma^*$ and outputs a string $H(s,x)\in \Sigma^{\ell(\sec)}$.
\end{itemize}
\end{definition}
A collision resistant hash function can be defined using the following experiment for a hash function $\Pi=(\GenH,H)$, an adversary $\A$, and a security parameter $\sec$:\\
\begin{mdframed}
\textbf{The collision-finding experiment} $\hashexperiment$:
\begin{enumerate}
\item
$s \leftarrow \GenH(1^\sec)$
\item
$(x, x') \leftarrow \A(s)$
\item
The output of the experiment is $1$ if and only if $\A$ successfully finds a collision, i.e.
\[ x\neq x' \mbox{ and } H(s, x)=H(s, x'). \] 
\end{enumerate}
\end{mdframed}
Then this experiment implicitly defines a collision resistant hash function.
\begin{definition}
\deflab{def:crhf}
\cite[Definition 4.12]{KatzL14}
A hash function $\Pi=(\GenH,H)$ is \emph{collision resistant} if for all probabilistic polynomial time adversaries $\A$ there exists a negligible function $\negl$ such that
\[\PPr{\hashexperiment=1}\le\negl(\sec).\]
\end{definition}

\subsection*{Background on $\delta$-Local Expander Graphs}
We now begin to describe the underlying DAGs in our code construction. 
We first define the class of graphs $\delta$-expanders and $\delta$-local expanders. 
\begin{definition}
Let $\delta > 0$. A bipartite graph $G=(U,V, E)$ is called a $\delta$-expander if for all $M_1\subseteq U$, $M_2 \subseteq V$ such that $|M_1|\ge\delta|U|$ and $|M_2|\ge\delta|V|$, there exists an edge $e\in M_1\times M_2$. 
\end{definition}

\begin{definition}
Let $\delta >0 $ and $G =([n], E)$ be a directed acyclic graph. 
$G$ is called a $\delta$-local expander if for every vertex $x \in [n]$ and every $r \le  \min\{ x, n-x\}$,  the bipartite graph induced by $U=[x,\ldots,x+r-1]$ and $V=[x+r,\ldots,x+2r-1]$ is a $\delta$-expander. 
\end{definition}
\thmref{thm:local-expander-exist}, due to Alwen \etal \cite{EC:AlwBloPie18}, states that $\delta$-local expanders exist with degree $\O{\log n}$.
\begin{theorem}\cite{EC:AlwBloPie18}
\thmlab{thm:local-expander-exist}
For any $n>0, \delta > 0$, there exists a $\delta$-local expander $G= ([n], E)$ with indegree $\O{\log n}$ and outdegree $\O{\log n}$.
\end{theorem} 
The construction of Alwen \etal \cite{EC:AlwBloPie18} is {\em probabilistic}. 
In particular, they show that there is a {\em random distribution} over DAGs such that any sample from this distribution is a $\delta$-local expander with high probability. 
The randomized construction of Alwen \etal \cite{EC:AlwBloPie18} closely follows an earlier construction of Erdos \etal \cite{ErdosGS75} \footnote{While Alwen \etal \cite{EC:AlwBloPie18} only analyze the indegree of this construction it is trivial to see that the outdegree is also $\O{\log n}$ since the construction of Erdos \etal \cite{ErdosGS75} overlays multiple bipartite expander graphs on top of the nodes $V=[n]$. Each bipartite expander graph has $\indeg, \outdeg \in\O{1}$ and each node is associated with $\O{\log n}$ expanders.}.  

We now list some properties of the $\delta$-local expander graphs shown in \cite{EC:AlwBloPie18}. 
We will use these properties to construct the $\Dec$ algorithm for the \CRLCC scheme.

\begin{definition}[$\alpha$-good node]
\deflab{def:cgood}
Let $G = ([n], E) $ be a DAG and let $S\subseteq [n]$ be a subset of vertices. Let $0 < c < 1$. 
We say $v \in [n]-S$ is $\alpha$-good under $S$ if for all $r > 0$:
\[ \left| [v-r+1,v]\cap S \right|\le \alpha r,\qquad \left| [v,v+r-1]\cap S \right|\le \alpha r.\]
\end{definition}
Intuitively, we can view the subset $S$ to be a set of ``deleted'' vertices. 
A vertex in the remaining graph is called $\alpha$-good if at most $\alpha$ fraction of vertices in any of its neighborhood are contained in the deleted set. 
In our case, we will ultimately define $S$ to be the nodes with the inconsistent labels.

The following result of \cite{EC:AlwBloPie18} shows that in any DAG $G$, even if we deleted a constant fraction of vertices, there still remain a constant fraction of vertices that are $\alpha$-good. 
 
\begin{lemma}[Lemma 4.4 in \cite{EC:AlwBloPie18}]
\lemlab{lem:cgood:most}
Let $G=([n], E)$ be a directed acyclic graph. For any $S \subset[n]$ and any $ 0< \alpha < 1$, 
at least $n-|S|\left(\frac{2-\alpha}{\alpha}\right)$ nodes in $G-S$ are $\alpha$-good under $S$.
\end{lemma}

\claimref{numReachable} rephrases a claim from \cite{EC:AlwBloPie18}. 
The only difference is that \cite{EC:AlwBloPie18} did not define the phrase ``G is a $\delta$-local expander around node $x$'' since the DAG $G$ they consider is assumed to satisfy this property for all nodes $x$. 
However, their proof only relies on the fact that  G is a $\delta$-local expander around node $x$.

\begin{claim}[\cite{EC:AlwBloPie18}] \claimlab{numReachable}
Let $G = (V=[n],E)$ be a $\delta$-local expander around $x$ and suppose that $x \in [n]$ is a $\alpha$-good node under $S \subseteq [n]$. Let $r >0$ be given. If $\alpha < \gamma/2$ then all but $2\delta r$ of the nodes in $I_r^*(x)\backslash S$ are reachable from $x$ in $G-S$. 
Similarly, $x$ can be reached from all but $2\delta r$ of the nodes in $I_r(x)\backslash S$. 
In particular, if $\delta < 1/4$ then more than half of the nodes in $I_r^*(x)\backslash S$ (resp. in $I_r(x)\setminus S$) are reachable from $x$ (resp. $x$ is reachable from) in $G-S$.  
\end{claim}

Alwen \etal \cite{EC:AlwBloPie18} also show that in a $\delta$-local expander graph $G$ any two $\alpha$-good vertices in $G-S$ are connected. 
As was the case for \claimref{numReachable} the only difference between our statement of \lemref{lem:cgood:connected} and \cite{EC:AlwBloPie18} is that \cite{EC:AlwBloPie18} did not define the phrase ``G is a $\delta$-local expander around node $x$'' since the DAG $G$ they consider is assumed to satisfy this property for all nodes $x$. 
However, their proof only relies on the fact that  G is a $\delta$-local expander around nodes $u$ and $v$.

\begin{lemma}[Lemma 4.3 in \cite{EC:AlwBloPie18}]
\lemlab{lem:cgood:connected}
Let $G=([n], E)$ be a $\delta$-local expander around nodes $u$ and $v$, $S \subseteq [n]$ and $0\le\alpha\le 1$. 
If $\delta<\min( (1-\alpha)/2,1/4)$, and nodes $u$ and $v$ are $\alpha$-good with respect to $S$ then  $u$ and $v$ are connected in $G-S$.
\end{lemma}

%% file: app-lcc-dec.tex
\section{Weak \CRLCC Decoder}
\subsection{Query Procedure for $i\ge8k - 4\ell(\sec)$}
\applab{app:i-ge-8k}
In this section, we describe the procedure for recovering a coordinate for input $i\ge8k-4\ell(\sec)$.  
We show that regardless of the underlying graph, any adversary that can change at most $\frac{\Delta_J k}{4}$ coordinates of a codeword obtained from $\Enc$, cannot prevent outputting $\bot$ or the correct bit for a query on the last $k' +1$ blocks. 

Consider the following algorithm $\Dec_{1}^w(s, i)$ for any $i \ge 8k-4\ell(\sec)$:  

\begin{mdframed}
$\Dec_{1}^w$: \underline{Input} : Index $i \in [n]$ such that $i\ge 8k-4\ell(\sec)$.
\begin{enumerate}
\item
Sample $\Theta\left(\frac{\log^3 k}{\ell(\sec)}\right)$ blocks $\{ w_{t} \}$ 
of $w$ for $t$ between $2k'$ and $3k'$ uniformly at random.
\item 
Decode each of the queried blocks $w_t$ to the corrected codeword,  $c_t := \ECC(\ECCD(w_t))$. (could possibly be a $\bot$ if $\ECCD$ fails to decode).
\item 
Let $c_{maj} = \text{majority} \{ c_t \}$ be the codeword of $\ECC$ which occurs majority of the times in Step (2) above. 
\item
Output symbol $(i \mod 4\ell(\sec))$ of $c_{maj}$.
\end{enumerate}
\end{mdframed}

We show in \lemref{lem:back} that $\PPr{\Dec_{1}^w(s, i) = c[i] } \ge 1-\negl(n)$ and uses at most $\O{\log^3 n}$ queries.

\begin{lemma}
\lemlab{lem:back}
Let $\Enc$ be as described in \secref{sec:lccencode}.
For any $ 8k-4\ell(\sec) \le i \le n$, $\Dec_{1}^w$ does the following: 
\begin{enumerate}
\item
For any $x \in \{0,1\}^k$ and $c =\Enc(s, x)$,  it holds that $\Dec_{1}^c(s, i) = c[i]$.
\item
For any $x \in \{0,1\}^k$, $c =\Enc(s, x)$  and $w \in \{0,1\}^n$ generated by any \PPT adversary such that $\HAM(c, w)\le\frac{\Delta_J k}{4}$, it holds that 
\[\PPr{\Dec_{1}^w(s, i) = c[i] } \ge 1-\frac{1}{n^{\log^2 n}}.\]
\end{enumerate}
Moreover, $\Dec_{1}^w$ makes at most $\O{\log^3 n}$ queries to $w$.
\end{lemma}
\begin{proof} First, note that for any codeword obtained from $\Enc$, $\Dec_{1}$ will always output the correct symbol. 
Let $w$ be the received word which is obtained by altering at most $\frac{\Delta_J k}{4}$ coordinates of some codeword $c$ obtained from $\Enc$, i.e., $0< \HAM(c,w) \le\frac{\Delta_J k}{4}$.
 From $\Enc$, we know that the last $k'$ blocks of any codeword are exactly the same.  
Since $\HAM(c,w) \le\frac{\Delta_J k}{4}$,  at most $\frac{1}{4}$ fraction of these blocks are modified in more than $\Delta_J$ fraction of their coordinates. 
Therefore, at least $\frac{3}{4}$ fraction of the last $k'$ blocks can be corrected to a unique codeword closest to them. 
Thus, each query finds an correct codeword block with probability at least $\frac{3}{4}$. 
Then by standard Chernoff bounds, the probability that the majority of the $\Theta\left(\frac{\log^3 k}{\ell(\sec)}\right)$ queries are correct is at least $1 - e^{-\Omega(n_s)}$, where $n_s = \Theta\left(\frac{\log^3 k}{\ell(\sec)}\right)$ is the number of sampled blocks.

Since $\Dec_{1}$ queries $\O{ \frac{\log^3 k}{\ell(\sec)}}$ blocks each of length $4 \ell(\sec)$ bits, the complexity of $\Dec_{1}$ is at most $\O{\log^3 k} = \O{\log^3 n}$. 
\end{proof}

\subsection{Query Procedure for $i<8k-4\ell(\sec)$}
\applab{app:i-le-8k}

In this section, we describe the algorithm for recovering a coordinate for input $i<8k-4\ell(\sec)$. 
Before we describe the algorithm, we need a few definitions and properties of the code we constructed in \secref{sec:lccencode}. 

Recall that for a codeword obtained from $\Enc$, the first $k'$ blocks of length $4\ell(\sec)$ each, corresponds to the encoding of the message symbols, and the next $k'$ blocks correspond to the encoding of labels of the nodes of a fixed $\delta$-local expander graph $G=([k'], E)$ generated using the $k'$ message blocks. 

Recall that we say that a node $v$ is {\em green} if it is locally consistent with the hash of its parents. 
In the next lemma, we describe an algorithm which verifies whether a given node of $G$ is green with respect to the labels obtained from the received vector $w \in \{0,1\}^n$. From the definition of a green node it follows that we need to query only the labels of the parents of a given vertex $v$ to verify if it is green. 
Since the indegree of $G$ is $\O{\log k'}$, we can check if a particular node is green by making at most $\O{\log k'}$ block queries to $w$. 
We now formalize the verification procedure.

\begin{lemma}
\lemlab{lem:green}
Let $G=([k'], E)$ be  $\delta$-local expander with indegree $\O{\log k'}$ used by $\Enc$ and let $w =(w_1, \ldots, w_{3k'}) \in \{0,1\}^{n}$ be the corrupted word obtained from the $\PPT$ adversary. 
There exists an algorithm \Green that uses $\O{\ell(\sec)  \log \left(n/\ell(\sec) \right)}$ coordinate queries to $w$ and outputs whether a given node $v \in [k']$ of $G$ is green or not. 
\end{lemma}
\begin{proof}
We claim the following algorithm achieves the desired properties:
\begin{mdframed}
\Green: \underline{Input}: Node $v \in [k']$ of $G$ with indegree $d=\O{\log{k'}}$.
\begin{enumerate}
\item
Query and decode the blocks $w_{k' + v_j}$ for all  $v_j \in \parents_{G}(v)$. Let $\ell_{k' + v_j}' :=\ECCD(w_{k' + v_j})$. Return `Red' if $\ECCD$ fails on any input. 
\item
Query and decode the block $w_{v}$. Let $x_{v}' :=\ECCD(w_{v})$. Return `Red' if $\ECCD$ fails. 
\item
If $\ECCD(w_{k' + v})= H(s, x_v' \circ w_{k' + v_1} \circ w_{k' + v_2}\circ\cdots\circ w_{k' + v_d})$, where $v_j \in \parents_{G}(v)$ for $j \in [d]$, then output `Green' 

\item Else output `Red'. 
\end{enumerate}
\end{mdframed}

$\Green$ outputs whether or not the label of node $v$, \emph{i.e.} $\ECCD(w_{k' + v})$, is green by querying at most $\O{\log k'}$ blocks of $w$. Since each block is $4 \ell(\sec)$ bits long, and $k' = k/\ell(\sec) = \O{n/ \ell(\sec)}$, the number of bits queried is equivalently $\O{\ell(\sec)  \log \left(n/\ell(\sec) \right)}$.
\end{proof}

Let $w \in \{0,1\}^n$ be the received word obtained from the $\PPT$ adversary. 
Let $S \subseteq [k']$ be the set of red nodes of $G$ with respect to its labeling obtained from $w$. 
Let $0 < \alpha < 1$. Recall from \defref{def:cgood} that we call a node $i \in [k']$ of $G$ to be $\alpha$-good under the set $S \subseteq [k']$ if for all $r > 0$, the number of nodes of $S$ in any $r$-sized neighborhood of $i$ is at most $\alpha\cdot r$. 
We now give an algorithm to locally verify if a given node of $G$ is $\alpha$-good with high probability. 

\begin{lemma}
\lemlab{lem:cgood:check}
For any $\alpha < \frac{3}{4}$, there exists a procedure that makes $\O{\ell(\sec)  \log^{3+\eps} \left(\frac{n}{\ell(\sec)} \right)}$ coordinate queries to $w$ and for any node $i \in [k']$ of $G$ does the following:
\begin{itemize}
\item Accepts if $i$ is $\frac{\alpha}{4}$-good under $S$ with probability  $1- \negl(n)$.
\item Rejects if $i$ is not $\alpha$-good under $S$ with probability  $1-\negl(n)$.
\end{itemize}
\end{lemma}
\begin{proof}
Consider the following algorithm $\Good$, where $S$ is defined to be the set of red nodes:
\begin{mdframed}
$\Good$: \underline{Input}:  Node $i \in [k']$ of $G$, $\alpha < \frac34 $, space parameter $\eps>0$. 
\begin{enumerate}
\item If $i$ is not \emph{green}, then Reject. 
\item For  every $p \in \{1, \ldots, \log k' \}$ do:
	\begin{itemize}
	\item Sample 
	$\log^{1+\eps} k'$ nodes $U_{1p} \subseteq [i-2^p + 1, i]$ and $U_{2p} \subseteq [i, i+2^p-1]$ with replacement. 
	\item If the fraction of red nodes in $U_{1p}$ or $U_{2p}$ is larger than $\frac{3\alpha}{8}$, then Reject.
	\end{itemize} 
\item Accept otherwise. 
\end{enumerate}
\end{mdframed}

Note that $\Good$ samples at most $\O{\log^{2+\eps} k'}$ nodes of $G$ and for each of those nodes, we check if it is green. Using Algorithm $\Green$ described in  \lemref{lem:green}, we can verify if a node is green using only $\O{\log k'}$ block queries to $w$. 
Therefore the number of coordinates of $w$ queried by $\Good$ is at most $\O{\ell(\sec)  \log^{3+\eps} \left(n/\ell(\sec) \right)}$.

To prove the correctness, first observe that if the node $i$ is red, then $i \in S$ and therefore by definition it is not $\alpha$-good for any $\alpha$. $\Good$ always rejects such nodes. 

Suppose node $i$ is $\alpha/4$-good. Then any $r$-sized neighborhoods of $i$, i.e $[i-r+1, i]$ and $[i, i+r-1] $ have at most $\alpha\cdot r/4$ red nodes. In particular, for all $p \in [\log k']$, the neighborhoods $U_{1p} = [i-2^p + 1, i]$ and $U_{2p} = [i, i+2^p-1]$ contain at most $\alpha /4$ fraction of red nodes. So, on sampling uniformly random $\log^{1+\eps} k'$ vertices from these intervals, the probability that we see more than $3\alpha/8$ fraction of red nodes is at most $e^{-\frac{\alpha}{12}\log^{1+\eps} k'}$ which follows from Chernoff bounds. Taking a union bound over all $\log k'$ intervals, $\Good$ accepts any $\alpha/4$-good node with probability at least $1 - e^{-\frac{\alpha}{24}\log^{1+\eps}(k')}$. Therefore $\Good$ accepts any $\alpha/4$-good node with probability at least $ 1 - \negl(k') = 1- \negl(n)$ for $\ell(\sec) = \polylog n$.

In order to show that the algorithm rejects any node which is not $\alpha$-good with high probability, we first show that any node which is not $\alpha$-good under $S$ is not $\alpha/2$-good for some interval of size $2^p$, $p > 0$.  The rejection probability then follows from Chernoff bounds. 
If node $i$ is not $\alpha$-good under $S$, then there exists a neighborhood $I$ of $i$ such that the number of red nodes in $I$ is at least $\alpha \lvert I \rvert$. 
Let $R(X)$ denote the number of red nodes for some subset $X\subseteq V$. 
Let $U_1 \subseteq I \subseteq U_2$ where $\lvert U_1 \rvert = 2^{p^*}$ and $\lvert U_2 \rvert = 2^{p^*+1}$ for some integer $p^* > 0$. 
Then
\[
R(U_2) \ge R(I) \stackrel{\mathclap{*}}{\ge} \alpha |I| \ge 
\alpha |U_1| = \alpha \cdot 2^{p^*} = \frac{\alpha}{2} \cdot 2^{p^* +1}.
\]
where the inequality (*) results from $i$ not being $\alpha$-good. 
So, the probability that this interval goes undetected by $\Good$ is at most $e^{-\frac{\alpha}{32} \log^{1+\eps} k'}$. Therefore, $\Good$ rejects any node which is not $\alpha$-good with probability at least $1 - \negl(k')$. 
\end{proof}

Now, we describe the decoding procedure for any index $i <8k-4\ell(\sec) $ using the $\Good$ procedure.  Let $(w, i, x)$ be the challenge provided by any \PPT adversary such that $\HAM( c , w) \leq \tau$, where $c = \Enc(x)$. 

Recall from the encoding procedure in \seclab{sec:lcc}, the first $8k-4\ell(\sec)$ bits correspond to either the the encoding of the message bits using $\ECC$ (the first $4k$ bits) or the encoding of the labels of the nodes the $\delta$-local expander (the next $4k$ bits). For any index query corresponding to the encoding of the label bits ($i > 4k$), we check if the label is tampered with or not. If it is not tampered, then the decoder returns the received bit, else it returns a $\bot$.  For any $i \leq 4k$, the decoder first tests if its corresponding label is tampered. If it is not tampered, then the decoder returns the corresponding received bit, else it returns a $\bot$. 

Recall that using Algorithm $\Dec_{1}$ we can successfully correct any index $i \geq 8k-4\ell(\sec)$. Therefore, to design $\Dec_{2}$, we assume the correctness of node $k'$ using $\Dec_{1}$. 
The procedure $\Dec_{2}$ first checks whether nodes $i$ and node $k'$ are $\alpha$-good for some $ \alpha < 3/4$. 
If they are both $\alpha$-good, then by \lemref{lem:cgood:connected} there is a path of green nodes connecting $i$ and node $k'$. 
Node $k'$ now serves as an anchor point by being both $\alpha$-good and correct due to the repetition code, so then \lemref{lem:cgood} implies node $i$ is also unaltered. We now formalize this and describe $\Dec_{2}$ in \lemref{lem:front}. 

\begin{lemma}
\lemlab{lem:front}
Let $\Enc$ be as described in \secref{sec:lccencode}.
For any $ 1 \le i < 8k-4\ell(\sec) $, there exists a procedure $\Dec_{2}$ with the following properties: 
\begin{enumerate}
\item
For any $x \in \{0,1\}^k$ and $c =\Enc(s, x)$,  $\Dec_{2}^c(s, i) = c[i]$.
\item
For any $x \in \{0,1\}^k$, $c =\Enc(s, x)$  and $w \in \{0,1\}^n$ generated by any \PPT adversary such that $\HAM(c, w)\le\frac{\Delta_J k}{4}$,
\[\PPr{\Dec_{2}^w(s, i) \in \{c[i], \bot \} } \ge 1-\negl(n).\]
\end{enumerate}
Moreover, $\Dec_{2}$ makes at most $\O{ \ell(\sec) \log^{3+\eps}\left( n/ \ell(\sec) \right)}$ queries to input $w$ for any constant $\eps>0$.
\end{lemma}
\begin{proof}
We claim the following procedure has the desired properties.
\begin{mdframed}
$\Dec_{2}^w$: \underline{Input}: Index $i<8k - 4\ell(\sec)$. 
\begin{enumerate}
\item Let $S$ be the set of red nodes of $G$ with respect to the labels obtained from $w$.
\item Let $ \alpha < \frac{3}{4}$.
\item Reconstruct the label of $k'$-th node, $\ell_{k', s}$ using a call to $\Dec_1$. 
\item If node $k'$ of $G$ with label $\ell_{k',s}$ and node $ \ceil{ \frac{i}{4 \cdot \ell(\sec)}} \mod k'$ are $\alpha$-good under $S$, then return $w_i$,
\item Else return $\bot$. 
\end{enumerate}
\end{mdframed}	
First we show that the query complexity of $\Dec_{2}$ is at most $\O{ \ell(\sec) \log^{3+\eps}\left( n/ \ell(\sec) \right)}$. 
Observe that  $\Dec_{1}$ described in \lemref{lem:back} uses majority decoding to reconstruct the entire block $w_{2k' + j}$, for any $ 0 \leq j \leq k'-1$. Moreover, it uses 
at most $\O{\log^3 n}$ coordinate queries to $w$ in order to reconstruct any of the last $k'$ blocks with $1 - \negl(n)$. Therefore we can assume that after Step (3), the label $\ell_{k',s}$ of $G$ is correct with very high probability. 
Also, using the procedure described in \lemref{lem:cgood:check}, we can check if both the nodes $k'$ and $i $ are $\alpha$-good using at most $\O{ \ell(\sec) \log^{3+\eps}\left( n/ \ell(\sec) \right)}$ queries for any $\eps >0$. Therefore the query complexity of $\Dec_{2}$ is $\O{ \ell(\sec) \log^{3+\eps}\left( n/ \ell(\sec) \right)}$. 

To show $(1)$, observe that if $w$ is a codeword produced by $\Enc$, then the label $\ell_{ k', s}$ is reconstructed correctly with probability $1$ by $\Dec_{1}$ in the first step. 
Also since $S = \emptyset$, all the nodes of $G$ are $\alpha$-good. 
Therefore, $\Dec_2$ always returns the correct codeword symbol.

Now, if $w$ is any string produced by a $\PPT$ adversary such that $0<\HAM(c,w)\le\frac{\Delta_J k}{4}$, then from the analysis of $\Dec_{1}$ in \lemref{lem:back}, we know that $\ell_{k',s}$ is reconstructed correctly with probability at least $1 - \negl(n)$. 
If both nodes $i$ and $k'$ of $G$ are $\alpha$-good under $S$, then by \lemref{lem:cgood:connected}, we know that there exists a path in $G$ from node $i$ to node $k'$ in $G$ which consists of only green nodes. 
So, node $i$ is a green ancestor of node $k'$. The correctness of $w_i$ then follows from \lemref{lem:cgood}. 
\end{proof}

~\\
We now combine the two correctors described above to obtain our local corrector $\Dec$ as follows: 
\begin{mdframed}
$\Dec^w$: \underline{Input} ~:  Index $i$.
\begin{enumerate}
\item If $i \geq 8k - 4\ell(\sec)$, return $\Dec_{1}^w(s, i)$
\item Else return $\Dec_{2}^w(s, i)$.
\end{enumerate}
\end{mdframed}

\begin{lemma}
\lemlab{lem:frontback}
Let $\Enc$ be as described in \secref{sec:lccencode}.
For any $ i \in [n]$, $\Dec$ does the following:
\begin{enumerate}
\item
For any $x \in \{0,1\}^k$ and $c =\Enc(s, x)$,  $\Dec^c(s, i) = c[i]$.
\item
For any $x \in \{0,1\}^k$, $c =\Enc(s, x)$  and $w \in \{0,1\}^n$ generated by any \PPT adversary such that $\HAM(c, w)\le\frac{\Delta_J k}{4}$,
\[\PPr{\Dec^w(s, i) \in \{c[i], \bot \} } \ge 1-\negl(n).\]
\end{enumerate}
Moreover, $\Dec$ makes at most $\O{ \ell(\sec) \log^{3+\eps}\left( n/ \ell(\sec) \right)}$ queries to input $w$ for any constant $\eps > 0$.
\end{lemma}
\begin{proof}
The proof follows from \lemref{lem:front} and \lemref{lem:back}. Since $n = 3k'\cdot 4\ell(\sec)$, any query on an index $i < 8k - 4\ell(\sec)$ is handled by \lemref{lem:front} while a query on an index $i \ge 8k - 4\ell(\sec)$ is handled by \lemref{lem:back}. Since each scenario uses at most $\O{ \ell(\sec) \log^{3+\eps}\left( n/ \ell(\sec) \right)}$ queries, the algorithm also uses $\O{ \ell(\sec) \log^{3+\eps}\left( n/ \ell(\sec) \right)}$ queries. Note that for $\ell(\sec) = \polylog n $, the query complexity of the decoder is $\polylog n$.
\end{proof}
\noindent
We now show that our construction $\Pi = (\Gen, \Enc, \Dec)$ is a \CRLCC scheme. 

\begin{proofof}{\thmref{thm:rlcc}}
$\Gen$ on input a security parameter  $\sec$ simulates the generator algorithm $\GenH$ of the \CRHF to output a random seed $s$. 

Consider $s \leftarrow \Gen(1^\sec)$, $\Enc$ described in \secref{sec:lccencode} and the decoder $\Dec$ described in \secref{sec:lccdecode} above. We claim that the triple $\Pi = (\Gen, \Enc, \Dec)$ is a \CRLCC scheme. 

From the construction of $\Enc$, we know that 
the block length of a codeword produced by $\Enc$ is $n = 3k'\cdot 4\cdot \ell(\sec) = 12k $.  Therefore, the information rate of the $\CRLCC$ is $\O{1}$.

From \lemref{lem:frontback} we know that on input $(w, i, x)$ generated by any \PPT adversary such that $\HAM(\Enc(x), w) \leq \frac{\Delta_J k}{4}$, $\Dec^w$ queries at most $\O{ \ell(\sec) \log^{3+\eps}\left( n/ \ell(\sec) \right)}$ coordinates of $w$ and returns $d \in \{ \Enc(x)[i] , \bot \}$ with probability at least $1 - \negl(n)$. 
Also, $\Dec$ on input any valid encoding $(\Enc(x), i, x)$ returns $\Enc(x)[i]$ with probability $1$. 
\end{proofof}

%% file: app-lcc-final.tex
\section{Strong $\CRLCC$ Decoder}
\applab{app:strong:dec}

Recall that the codeword $c = \Enc(x)$ consists of 3 parts of $t$ blocks each. 
The first $k/R$ bits (first $t$ blocks) correspond to the message symbols $x$ passed $m \cdot \beta \ell(\sec)$ bits at a time through $\ECC$ of rate $R$. 
The next $k/(\beta R)$ bits (second set of $t$ blocks) correspond to $t$ codewords $\ECC(U_j)$, and finally the last $k/(\beta R)$ bits (third set of $t$ blocks) correspond to the repetitions of the final \ECC codeword, $\ECC(U_{t})$.

 \noindent {\bf Note: } For simplicity of presentation, we let $\beta=1$ and $R=1/4$ for the discussion of the decoder in this section. The proof extends similarly for other values of the parameters $\beta$ and $R$.

Observe that at the end of the construction, each meta-node contains exactly one node ({\emph i.e.} the final node) with indegree more than $2$ (it has indegree $m$) and outdegree more than $2$ (it has outdegree $\O{m}$). 

Similar to the decoder $\Dec$ of \secref{sec:lccdecode},  the current decoding procedure has separate subroutine to handle indices corresponding to the first $2t-1$ blocks and a separate decoding procedure for the remaining indices. 
In the following, we define $\metanode{i}$ to be the meta-node corresponding to index $i$. 
That is,
\[\metanode{i} =
\begin{cases}
\ceil{\frac{i}{4m \ell(\sec)}},& i \leq 4k \\
\ceil{ \frac{i - 4k}{ 4m\ell(\sec)} },& 4k < i \le 8k,\\
t, & i > 8k
\end{cases}\]

Before we describe the decoder for the Strong \CRLCC, we need a few definitions and supporting lemmas.

The following claim follows from the construction of the metagraph $G_0$ which was shown in \cite{EC:AlwBloPie18}. 
The bounded outdegree allows us to test meta-nodes with a bounded (in this case logarithmic) number of queries.
\begin{claim}
$\outdeg(G_0)=\Theta(\log k)$.
\end{claim}

Let $c  = \left(c_1 \circ \ldots \circ c_{3t} \right)= \Enc(s,x)$. Recall that for a meta-node $u \in [t] = V(G_0)$ we have $c_u = \ECC(x_{(u-1)m +1} \circ \ldots \circ x_{um})$ and $c_{u+t} = \ECC(\ell_{u_1,s} \circ \ldots \circ \ell_{u_m,s})$. Given a (possibly) corrupted $c'$ codeword we let $\tilde{c}_i = \ECC(\ECCD(c_i'))$. 
We note that if $c_i'$ and $c_i$ are sufficiently close then we will always have $\tilde{c}_i  = c_i$. If  $\tilde{c}_i  \neq c_i$ we say that the block $i$ was tampered and if $\tilde{c}_i  = c_i$ we say that the block is {\em correct}. 
We let $T \subseteq [t-1]$ denote the set of all nodes $u < t$ (excluding the last node $t$) with the property that either blocks $u$ or ${u+t}$ were {\em tampered}. 
We also let $\left(\ell_{u_1,s}' \circ \ldots \circ \ell_{u_m,s}'\right) = \ECCD(c_{u+t})$ and  $\left(x_{(u-1)m +1} \circ \ldots \circ x_{um}\right) = \ECCD(c_u)$. 
We say that a node $u_i \in V(G)$ with $\parents(u_i) = u_i^1,\ldots u_i^d \in V(G)$ is green if $\ell_{u_i,s}' = H\left(s , x_{(u-1)m+i} \circ \ell_{u_i^1,s}' \circ \ldots \circ \ell_{u_i^d,s}'\right)$ i.e., the given label of node $u_i$ is consistent with the labeling rule given the data value $x_{(u-1)m+i}$ and the given labels of $u_i$'s parents. 
We say that an edge $(u_m,v_i) \in E(G)$ is green if both $u_m$ and $v_i$ are green. 

\begin{definition}
\deflab{def:meta:green}

We call a meta-node $u \in V(G_0)$ \emph{green} if {\em both}:
\begin{itemize}
\item
At least $\frac{2}{3}$ fraction of the corresponding nodes $u_1,\ldots,u_m \in V(G)$ are are green. 
\item
The final node $u_m$ in the meta-node is also green.
\end{itemize}
We call a meta-node \emph{red} if the meta-node is not green. 
Similarly, we call a meta-node $u$ \emph{correct} if we are able to correctly decode the blocks $c_u$ and $c_{u+t}$. 

\end{definition}
We stress that a \emph{red} meta-node is not \emph{necessarily} tampered nor is a \emph{green} meta-node \emph{necessarily} correct. 

\subsection{Query Procedure for $\metanode{i}=t$}
\seclab{sec:strongdec-1}

In the case that $\metanode{i}=t$, we run a procedure $\Dec^w_{= t}(s, i)$ that takes advantage of the fact that any adversary that can change at most $\frac{\Delta_J k}{\error}$ coordinates of any codeword obtained from $\Enc$, 
cannot prevent local correction for a query on the last $t$ blocks corresponding to the repetition of $\ECC$ codeword. 
Similar to $\Dec_{1}^w$ in \secref{sec:lccdecode}, the algorithm $\Dec^w_{= t}(s, i)$ randomly samples a number of blocks, performs decoding using $\ECCD$ and takes the majority of the samples. 
\begin{lemma}
\lemlab{lem:maj:back}
Let $\Enc$ be as described in \secref{sec:maj:lccencode}.
For any $i \in [n]$ such that $ \metanode{i} = t$, there exists an algorithm $\Dec^w_{= t}$ that does the following: 
\begin{enumerate}
\item
For any $x \in \{0,1\}^k$ and $c =\Enc(s, x)$,  $\Dec^c_{= t}(s, i) = c[i]$.
\item
For any $x \in \{0,1\}^k$, $c =\Enc(s, x)$  and $w \in \{0,1\}^n$ generated by any \PPT adversary such that $\HAM(c, w)\le\frac{\Delta_J k}{\error}$,
\[\PPr{\Dec^w_{= t}(s, i) = c[i] } \ge 1-\frac{n}{n^{\log n}}.\]
\end{enumerate}
Moreover, $\Dec^w_{= t}$ makes at most $\O{\ell(\sec) \cdot \log^{2+\eps} n}$ coordinate queries to input $w$.
\end{lemma}
\begin{proof}
Consider the following algorithm $\Dec^w_{= t}(s, i)$ for any $i \in [n]$ with $\metanode{i}= t$: 
\begin{mdframed}
$\Dec^w_{= t}(s, i)$: \\
\underline{Input} ~: Index $i \in [n]$ with $\metanode{i}= t$.
\begin{enumerate}
\item Sample (with replacement) $s:= \O{\log^{1+\eps} n}$ blocks $\{ \tilde{y_1}, \ldots, \tilde{y_s} \}$ each of length $4m\ell(\sec)$ from the last $t$ blocks. 
\item
Let $y_j = \ECC(\ECCD(\tilde{y_j}))$ for every $j \in [s]$ to correct each of the sampled blocks. 
\item
Output majority of $ \{ y_{j, \hat{i}} \mid j \in [s] \}$, where $\hat{i}=i \pmod{4m \ell(\sec)}$ and $y_{j, \hat{i}}$ is the $\hat{i}$-th coordinate in $y_j$.
\end{enumerate}
\end{mdframed}
First, note that for any codeword obtained from $\Enc$, $\Dec_{= t}$ will always output the correct symbol.

Let $w$ be received word which is obtained by altering at most $\frac{\Delta_J k}{\error}$ coordinates of some codeword i.e, $0< \HAM(c,w) \le\frac{\Delta_J k}{\error}$ for some codeword $c$ obtained from $\Enc$. 

We know that the last $t$  blocks of any codeword are exactly the same, \ie, $c_{2t} = c_{2t+1} = \cdots = c_{3t}$.

Since $\HAM(c,w) \le\frac{\Delta_J k}{\error}$, at most $\frac{t}{4\error}$ blocks are corrupted in more than $\Delta_J$ fraction of coordinates. One can therefore decode all blocks among the last $t$ blocks correctly except for those $\frac{t}{4\error}$ blocks with large corruption. 

Since we query blocks uniformly at random, each query finds a block with large corruption with probability at most $\frac{1}{4\error}$. Therefore, the probability that majority of the queried blocks are corrupted beyond repair is at most $\left( \frac{1}{4\error} \right)^{s/2} = \negl(n)$, for $s=\O{\log^{1+\eps} n}$.

Since $\Dec_{= t}$ queries $\O{\log^{1+\eps} n}$ blocks each of length $4m \cdot \ell(\sec)$ and  
$m=\O{\log n}$, then at most $\O{\ell(\sec) \cdot \log^{2+\eps} n}$ queries are made in total.
\end{proof}

\subsection{Query Procedure for $\metanode{i}<t$}
\seclab{sec:strongdec-2}

Recall from the construction that $G_0$ is the metagraph of $G$. 
We define the green subgraph of $G_0$ to be the subgraph consisting of exactly the edges whose endpoints in $G$ are both green nodes. 
We show there exists an efficient procedure to check whether an edge is in the green subgraph, which we shall use to determine if nodes in $G$ have $\delta$-local expansion in the green subgraph. 
If there is local expansion around nodes $u$ and $v$ in the green subgraph, then without loss of generality, there exists a path from $u$ to $v$ that only consists of green nodes by \lemref{lem:greenconnected}.
Hence if $v$ is correct, then we show in \lemref{lem:connected:correct} that either $u$ is also correct, or a hash collision has been found. 
Nevertheless, this procedure is vacuous if there is a low number of nodes $u,v$ with local expansion in the green subgraph. 
To address this concern, we show in \lemref{lem:alpha:expansion} that for a specific range of $\alpha$, if a node $u$ is $\alpha$-good with respect to a tampered set, then there must also be $2\delta$-local expansion around $u$ in the green subgraph. 
 
In this section, we design the corrector for indices $i \in [n]$ such that $\metanode{i}<t$.

Recall that we call a meta-node green if most of its nodes including the final node are green.
First, we claim the existence of a procedure to test whether a meta-node is green. 
\begin{lemma}
\lemlab{lem:meta:green}
Let $G_0$ be the meta-graph with $t$ vertices used by $\Enc$ and let 
$w = (w_1, \cdots, w_{3t})\in \{ 0,1\}^n$ be the corrupted word obtained from the $\PPT$ adversary. 
There exists a procedure $\GreenMeta$ that uses $\O{\ell(\sec) \cdot \log^2 n }$ coordinate  queries to $w$ and checks whether a meta-node $u$ of $G_0$ with labels corresponding to $w$ is green. 
\end{lemma}

\begin{proof}
We claim the following procedure satisfies the desired properties:
\begin{mdframed}
$\GreenMeta(u,w)$: \\ \underline{Input} ~: Meta-node index $u \in [t]$, corrupted codeword $w \in \{0,1\}^n$ 
\begin{enumerate}
\item
Check if the final node in meta-node $u$ is green:
\begin{itemize}
\item
Query coordinates of $w_{t + u}$ and retrieve $(\ell_{u_1}', \ldots, \ell_{u_{m}}') := \ECCD(w_{t + u})$. 
\item
Query the coordinates of $w_{u}$ and retrieve $(x_{u_1}', \ldots, x_{u_m}') := \ECCD(w_{u})$.
\item
Check whether $\ell_{u_m}' = H(s, x_{u_m}' \circ \ell_{u_1}' \circ \cdots \circ \ell_{u_{m-1}}' )$ for the final node $u_m$ in meta-node $u$, where $u_1,\ldots, u_{m-1}$ are the parents of $u_m$, which are all in meta-node $u$.
\item
If $\ell_{u_m}'$ is not consistent, then output `Red'.
\end{itemize}
\item
For each node $u_j$ in meta-node $u$, check whether $u_j$ is green:
\begin{itemize}
\item
Let $u_{j-1}$ (in meta-node $u$) and $p_r$ ($r$-th node in some meta-node $p$) be the parents of $u_j$.
\item
Query $w_{t + p}$ and retrieve $\ell_{p_r}'$ using $\ECCD(w_{t + p})$. 
\item
Check whether $\ell_{u_j}'=H(s, x_{u_j}' \circ \ell_{u_{j-1}}' \circ \ell_{p_r}' )$.
\end{itemize}
\item
If at least $\frac{2}{3}$ fraction of the nodes in meta-node $u$ are green, then output `Green'.
\item
Else, output `Red'.
\end{enumerate}
\end{mdframed}
Recall from \defref{def:meta:green} that a green meta-node first requires that at least $\frac 23$ of the underlying nodes in meta-node $u$ are green. 
Since a node is green if its label is consistent with the labels of its parents, then the procedure $\GreenMeta$ properly recognizes whether at least $\frac{2}{3}$ of the underlying nodes in $u$ are green after decoding using $\ECCD$.
Secondly, a green meta-node requires that the final node in $u$ is green. 
In this case, the final node $u_m$ has $m-1$ parents, and again the procedure recognizes whether $\ell_{u_m}' = H(s, x_{u_m}' \circ \ell_{u_1}' \circ \cdots \circ \ell_{u_{m-1}}' )$  after obtaining the labels using $\ECCD$.

Observe that in Step (1), we query two blocks of $w$ each of length $4m\ell(\sec) $ bits.  Also, in Step (2), for each of the $m-1$ nodes $u_j$ in the meta-node $u$ we query one additional block of $w$ corresponding to the parent meta-node of $u_j$. Since each block is of length $4m\ell(\sec)$ bits long, the total query complexity is $\O{m^2 \ell(\sec)}  = \O{\ell(\sec) \cdot \log^2 n }$ for $m=\O{\log n}$. 
\end{proof}

We now define the notion of a green edge in the meta-graph $G_0$.
 
Let $u$ and $v$ be meta-nodes which are connected in $G_0$. From the construction of the degree reduced graph $G$, we know that there exists an edge from the $m$-th node of the meta-node $u$ to some node $v_u$ of the meta-node $v$.
We say that an edge from meta-node $u$ to meta-node $v$ is green if the corresponding end points in $G$ are green, \ie, the node $u_m$ of meta-node $u$ and the node $v_u$ in meta-node $v$ are green.  
We say an edge is \emph{red} if it is not green. 
Note that it is possible for the end node $v_u$ to be red even if the meta-node $v$ is green. 

Given any $w \in \{0,1\}^n$, we now describe a procedure $\GreenEdge(u, v, w)$ which locally verifies whether a given edge $(u,v)$ of $G_0$ is green or not.  

\begin{lemma}
\lemlab{lem:edge:green}
Let $G_0$ be the meta-graph with $t$ vertices used by $\Enc$ and let 
$w = (w_1, \cdots, w_{3t})\in \{ 0,1\}^n$ be the corrupted word obtained from the $\PPT$ adversary. 
There exists a procedure $\GreenEdge$ that makes $\O{ \ell(\sec) \log(n)}$ coordinate queries to $w$ and checks whether an edge $(u,v)$ of $G_0$ with labels corresponding to $w$ is green. 
\end{lemma}

\begin{proof}
\begin{mdframed}
$\GreenEdge(u,v,w)$: \\ \underline{Input} ~: Meta graph edge $(u,v)$, corrupted codeword $w \in \{0,1\}^n$ 
\begin{enumerate}
\item Check if the final node in meta-node $u$ is green:
\begin{itemize}
\item
Query coordinates of $w_{t + u}$ and retrieve $(\ell_{u_1}', \ldots, \ell_{u_{m}}') := \ECCD(w_{t + u})$. 
\item
Query the coordinates of $w_{u}$ and retrieve $(x_{u_1}', \ldots, x_{u_m}') := \ECCD(w_{u})$.
\item
Check whether $\ell_{u_m}' = H(s, x_{u_m}' \circ \ell_{u_1}' \circ \cdots \circ \ell_{u_{m-1}}' )$ for the final node $u_m$ in meta-node $u$, where $u_1,\ldots, u_{m-1}$ are the parents of $u_m$, which are all in meta-node $u$.
\item
If $\ell_{u_m}'$ is not consistent, then output `Red' edge.
\end{itemize}
\item Check if the node $v_u$ in meta-node $v$ is green:
\begin{itemize}
\item
Query coordinates of $w_{t + v}$ and retrieve $(\ell_{v_1}', \ldots, \ell_{v_{m}}') := \ECCD(w_{t + v})$. 
\item
Query the coordinates of $w_{v}$ and retrieve $(x_{v_1}', \ldots, x_{v_m}') := \ECCD(w_{v})$.
\item
Check whether $\ell_{v_u}' = H(s, x_{v_u}' \circ \ell_{v_{u-1}}' \circ \ell_{u_m}' )$ 
\item
If $\ell_{v_u}'$ is not consistent, then output `Red' edge.
\end{itemize}
\item Output 'Green' edge.
\end{enumerate}
\end{mdframed}

The procedure $\GreenEdge$ first checks if the node $u_m$ of meta-node $u$ is green, and then checks if the connecting node $v_u$ of meta-node $v$ is green. The correctness of $\GreenEdge$ then follows from the definition. 
Note that checks (1) and (2) both query $2$ blocks of $w$, each of length $4 m \ell(\sec)$. Hence the query complexity of the $\GreenEdge$ is $\O{m \ell(\sec)} = \O{\ell(\sec) \log(n)}$, for $m = \O{\log(n)}$. 

\end{proof}

Using the definition of a green edge, we now define a green subgraph of $G_0$ and the notion of local expansion which is an essential ingredient to describe the decoder. 

\begin{definition}\deflab{def:greengraph}
The green subgraph of $G_0$, denoted by $G_g$, is the subgraph of $G_0$ that contains all green edges of $G_0$. 
\end{definition}

\begin{definition}\deflab{def:localexpander}
We say that  a DAG $G$ has $\delta$-local expansion {\em around a node} $v\in V$ if for all $r$:
\begin{enumerate}
\item
$A=[v,v+r-1]$, $B=[v+r,v+2r-1]$ and $W\subseteq A$, $X\subseteq B$ with $|W|,|X|\ge\delta r$, then there exists an edge between $W$ and $X$
\item
$C=[v,v-r+1]$, $D=[v-r,v-2r+1]$ and $Y\subseteq C$, $Z\subseteq D$ with $|Y|,|Z|\ge\delta r$, then there exists an edge between $Y$ and $Z$
\end{enumerate}
\end{definition}
Recall that we say that a DAG $G$ is a $\delta$-local expander if all nodes $v\in V$ have $\delta$-local expansion.

A key step in designing the decoder is to verify whether $G_g$ has local expansion property around some given meta-node. Before describe the verification procedure and why it is necessary, we introduce  notations and list some important properties of the meta-graph $G_0$ and its green subgraph $G_g$. 

Let $u$ be any meta-node of $G_0$.  Let $A^{u, r} := [u, u+r-1]$ and $B^{u,r} := [u+r, u+2r-1]$  and $r < \frac 12 (t- u)$, then from the definition of a $\delta$-local expander, we know that $G_0$ contains $\delta$-expander between $A^{u, r}$ and $B^{u, r}$. 
Let us denote this subgraph by $H^{u, r}$, i.e., $H^{u,r} = (A^{u, r} \cup B^{u, r}, (A^{u, r} \times B^{u, r}) \cap E)$. We assume the following property of $G_0$.  

\begin{fact}\cite{EC:AlwBloPie18}
$H^{u, r}$ has a constant indegree denoted by $d_{\delta} \in \O{1}$.
\end{fact}

Define $H_g^{u,r}$ to be the green subgraph of $H^{u, r}$, \ie, the subgraph of $H^{u, r}$ restricted to $G_g$.
\begin{lemma}\lemlab{lem:exp:red}
Let $n_r$ be the number of red edges in $H^{u, r}$. If $H_g^{u,r}$ is not a $4 \delta$-expander,  then $n_r \geq 3 \delta r$. 
\end{lemma}
\begin{proof}
If $H_g^{u,r}$ is not a $4 \delta$-expander, then there exist subsets $X \subseteq A^{u, r}$, $Y\subseteq B^{u,r}$ with $|X|,|Y| \ge 4\delta r $ such that there are no edges between $X$ and $Y$ in $H_g^{u,r}$. Since $H^{u,r}$ is a $\delta$-expander, at most $\delta r$ nodes $Y' \subset Y$  are not connected to any node in $X$. There are at most $|Y-Y'| \ge 3\delta r$ nodes which are all connected to $X$. So the number of edges in $H^{u,r}$ are at least $3\delta r$ all of which are red. 
\end{proof}

\begin{remark}
Though the observation states the property for descendants of a meta-node $u$, i.e., $A^{u, r} = [u, u+r-1]$ and $B^{u,r}=[u+r, u+2r-1]$, we note that it also holds for ancestors as well, \ie,  $A^{u, r} = [u-r +1, u]$ and $B^{u,r}=[u-2r-1, u-r]$. 
\end{remark}

Using the key observation described above, we now give an efficient procedure to verify if the green meta-graph $G_g$ has $4\delta$-local expansion around a given meta-node $u$. The idea is to estimate the number of red edges in increasing intervals around the given meta-node  and reject if this estimate is large for any interval.

\begin{lemma}
\lemlab{lem:green:expander}
Let $G_0$ be the meta-graph with $t$ vertices used by $\Enc$ and let 
$w = (w_1, \cdots, w_{3t})\in \{ 0,1\}^n$ be the corrupted word obtained from the $\PPT$ adversary. There exists a procedure $\DeltaExpansion$ that makes $\O{m \ell(\sec) \log^{2+\eps} t}$ coordinate queries to $w$ and rejects any meta-node $u$ around which  $G_g$ does not have $4 \delta$-local expansion with probability $1-\negl(n)$. 
\end{lemma}

\begin{proof}
Consider the following algorithm $\DeltaExpansion$: 
\begin{mdframed}
$\DeltaExpansion(u,w, \delta)$: \\ \underline{Input} ~: Green meta node $u$, corrupted codeword $w \in \{0,1\}^n$
\begin{enumerate}

\item For every $p \in \{1, \ldots, \log t\}$: 
	\begin{itemize}
	\item  Consider a bipartite expander $H^{u, 2^p}$, with node sets $A_p = [u, u+2^p-1]$ and $B_p=[u+2^p,u+2^{p+1}-1]$.
	\item Randomly sample $ s:= \log^{1+\eps} t$ edges from the subgraph with replacement.
	\item Count the number of red edges in the sample. 
	\item If the number of red edges in the sample is larger than $\frac52 \delta s$, then 'Reject'.
	\item Consider a bipartite expander $H^{' u, 2^p}$ with node sets $A'_p = [u-2^p+1, u]$ and $B'_p=[u-2^{p+1} + 1,u-2^{p}]$.
	\item Randomly sample $ s:= \log^{1+\eps} t$ edges from the subgraph with replacement.	\item Count the number of red  edges in the sample. 
	\item If the number of red edges in the sample is larger than $\frac52 \delta s$, then 'Reject'.
	\end{itemize}
\item Accept otherwise. 
\end{enumerate}
\end{mdframed}
Note that $\DeltaExpansion$ calls the procedure $\GreenEdge$ for every sampled edge to check whether it is red or green. 
Each call to $\GreenEdge$ requires $\O{m \ell(\sec)}$ coordinate queries to $w$. 
The total number of calls to $\GreenEdge$ is at most $2 \log^{2+\eps} t$ since we sample at most $\log^{1+\eps} t$ distinct edges from each of the $\log t$ subgraph $H^{u, 2^p}$. 
Therefore, the total query complexity of $\DeltaExpansion$ is upper bounded by $\O{m \ell(\sec) } \log^{2+\eps} t$. 

Suppose $u$ is a meta-node around which $G_g$ does not have $4\delta$-local expansion. 
Then there exists a fixed node set $[u, u+2^p-1]$ and $[u+2^p, u+2^{p+1}-1]$ whose corresponding subgraph has at least $3\delta2^p$ red edges, by \lemref{lem:exp:red}.  
Therefore on sampling uniformly at random with replacement, we expect to see at least $3\delta s$ red edges. 
So the probability that we see at most $\frac{5}{2}\delta s$ red edges is at most $\exp(-\O{\log^{1+\eps} t})$, from standard Chernoff bounds.

Note that we accept $u$ if we see at most $\frac{5}{2}\delta s$ red edges in all the $2\log t$ intervals. 
By taking union bound the probability that we accept a meta-node around which $G_g$ does not have $4\delta$-expansion is at most $\frac{2\log t}{\exp(-\O{\log^{1+\eps} t})}$.
\end{proof}

Now, we list some key properties of the meta-nodes around which $G_g$ has local expansion. 
These properties are important to understand why we needed the local expansion verification procedure. 
Essentially we show that if any meta-node is $\alpha$-good under the set of tampered meta-nodes, then $G_g$ has $2\delta$-local expansion about it (for appropriately chosen value of $\alpha$). 
Then we show that any two meta-nodes around which $G_g$ has local expansion are connected by a path of green nodes in $G$. 
Now similar to the weak \CRLCC decoder, we can argue about the correctness of any ancestor of a correct node in a green path (unless the adversary finds a hash collision). 

Recall that we call a meta-node $u$ to be $\alpha$-good under the $T$ if every interval for every interval $[u, u+r-1 ] \cap T \le \alpha r$.   We see that if a meta-node $u$ is $\alpha$-good under $T$, then $G_g$ has $2\delta$-local expansion around $u$.

\begin{lemma}\lemlab{lem:alpha:expansion}
Let $u$ be a meta-node that is $\alpha$-good with respect to $T$. 
Then there is $2\delta$-local expansion around $u$ in $G_g$ for any $\alpha<\frac{\delta}{2}$.
\end{lemma}
\begin{proof}
Given any $r>0$, consider the intervals $A=[u, u+r]$ and $B=[u+r+1, u+2r]$. 
Recall that $G_0$ is a $\delta$-local expander.
Let $X\subseteq A$ and $Y\subseteq B$ be given sets with size $2\delta r$ each. 
Since $u$ is $\alpha$-good with respect to $T$, then $|T\cap X|\le 2\alpha r$ and $|T\cap Y|\le 2\alpha r$. 
Thus, $|X-T|\ge(2\delta-2\alpha)r$ and similarly, $|Y-T|\ge(2\delta-2\alpha)r$. 
Setting $\alpha<\frac{\delta}{2}$ shows that both $X-T$ and $Y-T$ contains at least $\delta r$ nodes, and so there exists an edge between $X-T$ and $Y-T$. 
Therefore, there is local expansion around $u$. 
\end{proof}

Moreover, we see that the number of red edges in the $2\delta$-expander is at most $3\alpha r d_{\delta}$. 
\begin{lemma}
If $u$ is a meta-node that is $\alpha$-good under $T$, then for any $r>0$ the number of red edges in between the intervals $[u, u+r-1]$ and $[u+r, u+2r-1]$ is at most $3\alpha r d_{\delta}$.
\end{lemma}
\begin{proof}
Since $u$ is $\alpha$-good under $T$, the intervals $A:=[u, u+r-1]$ and $B:=[u+r, u+2r-1]$ contain at most $\alpha r$ and $2 \alpha r$ tampered nodes respectively. Since we assumed that the $\delta$-expander contained in $A$ and $B$ has constant degree $d_{\delta}$, the number of red edges is upper bounded by $3\alpha r d_{\delta}$.
\end{proof}

So in order to ensure that this number is not too large (for \DeltaExpansion to succeed), we choose $\alpha$ such that $3\alpha r d_{\delta} \ll  3\delta r$. This guarantees that the procedure \DeltaExpansion definitely accepts nodes that are $\alpha$-good under $T$. 

We now show that any two meta-nodes with local expansion property are connected by a path of green edges in $G_g$. This also ensures that the corresponding nodes are connected by a path of green nodes in the underlying graph $G$.

\begin{lemma}\lemlab{lem:greenconnected}
If $G_g$ has $4\delta$-local expansion around the meta nodes $u$ and $v$, then there exists a path from $u$ to $v$ in $G_g$.
\end{lemma}
\begin{proof}
Let $r$ be such that $v = u + 2r-1$. Consider the $4\delta$ expander between $A:=[u, u+r-1] $ and $B:=[u+r, u+2r-1]$. 
We show using Claim~\ref{claim:expander:reachable} that $u$ is connected to at least $4\delta r$ meta-nodes in $A$, and similarly $v$ is reachable from at least $4\delta r$ meta-nodes in $B$. 
Now since $G_g$ has $4\delta$ expansion around $u$, there exists an edge between any two large enough subset of nodes of $A$ and $B$. 
Hence, the set of nodes reachable to $u$ in $A$ are connected to the set of nodes which are reachable from $v$ in $B$. 

It now remains to show that $u$ and $v$ are connected to at least $4 \delta r$ meta-nodes in $A$ and $B$ respectively. 
Let $i$ be such that $2^i \leq r < 2^{i+1}$. 
From Claim~\ref{claim:expander:reachable}, we get that $u$ is connected to at least $\frac 34 \cdot 2^i = \frac 38 \cdot 2^{i+1} \geq 4\delta r$ (for $\delta < \frac{1}{16}$) meta-nodes in the interval $[u, u+r-1]$. 
Using a similar argument we conclude that $v$ is also connected to at least $4 \delta r$ nodes in $[v-r+1, v] = [u+r, u+2r-1]$.

We now prove Claim~\ref{claim:expander:reachable}.
\begin{claim}\label{claim:expander:reachable}
For any $i > 0$, the number of meta-nodes in the interval $[u, u+2^i-1]$ reachable from $u$ is at least $\frac 34 \cdot 2^i$. 
Similarly, the number of meta-nodes in the interval $[u-2^i+1, u]$ reachable from $u$ is at least $\frac 34 \cdot 2^i$
\end{claim}
\begin{proof}
We prove this claim by induction on $i$. Let $R_i(u)$ denote the set of meta-nodes in $[u, u+2^{i}-1]$ reachable from $u$. 
We have to show that $| R_i(u) | \ge \frac 34 \cdot 2^i$. 

In the base case, for $i=0$, we know from $4\delta$-expansion of $[u, u+1]$ around $u$ that there is at least one edge between $u$ and $u+1$ in $G_g$. Hence, $| R_0(u) | \ge 1$. 

Let us assume the induction hypothesis for all $i \le i_0$. To prove the induction step for $i=i_0+1$, consider the intervals $A:=[u, u+2^{i_0} -1]$ and $B:=[u+2^{i_0}, u+2^{i_0+1}-1]$. 
Let $NR$ denote the set of non-reachable meta-nodes in $B$ not reachable from $A$. 
We know that $| NR | < 4\delta \cdot 2^{i_0}$. 
This follows from the fact that $G_g$ has $4\delta$-local expansion around $u$, and hence if $|NR| \ge 4\delta \cdot 2^{i_0}$, then there would be a an edge between $R_{i_0}(u)$ and $NR$ which would contradict the definition of $NR$. 
Therefore, 
\begin{align*}
R_{i_0+1}(u) &\geq R_{i_0}(u) + 2^{i_0} - | NR | \\
&\geq \frac34 \cdot 2^{i_0} + 2^{i_0} - 4\delta \cdot 2^{i_0} ~~\mbox{ (from IH) } \\
&\geq \frac34 \cdot 2^{i_0 +1} ~~\mbox{ (for $\delta < \frac{1}{16}$)}.
\end{align*}
\end{proof}
\end{proof}

\begin{lemma}\lemlab{lem:connected:correct}
Suppose there exists a path from meta-node $u$ to $v$ in the green subgraph $G_g$. If $v$ is untampered (correct), then either  $u$ is also untampered or the adversary has found a hash collision. 
\end{lemma}
\begin{proof}
The proof of this lemma follows from the proof of \lemref{lem:cgood} once we show that there exists a path in $G$ with all green nodes from any node in the meta-node $u$ to some node in the meta-node $v$. 

Consider the path from meta-nodes $u$ to $v$ in $G_g$. 
Any two adjacent edges $(p, q)$ and $(q,r)$ in $(u,v)$ path in $G_g$ corresponds to the edge from the last node $p_m$ of $p$  to some node $q_i$ $(i < m)$ of meta-node $q$, and an edge from the last node $q_m$ of $q$  to some node $r_j$ $(j < m)$ of meta-node $r$. 
Since these edges are green, we know that the nodes $p_i, p_m$ and $q_j$ are green. 
Now, from the construction of the graph $G$, we conclude that there is an edge from $p_i$ to $p_m$. 
Therefore, there is a path in $G$ from any node in the meta-node $u$ to some node in the meta-node $v$ with all green nodes. 
\end{proof}

Equipped with all the necessary procedures, we now present the decoder for the Strong \CRLCC for any coordinate query $i \in [n]$ such that $\metanode{i} < t$.  

\begin{mdframed}
$\Dec_{< t}^w$: \underline{Input} ~: Index $i \in [n]$ with $u:= \metanode{i}< t$ of $G$, $\frac{\delta}{20d_{\delta}} \leq \alpha \leq \frac{\delta}{10d_{\delta}}$, $\delta < \frac{1}{16}$ space parameter $\eps>0$. 

\begin{enumerate}
\item Use $\DeltaExpansion$ to verify if $G_g$ has $4\delta$-local expansion around $u$.
\item If $\DeltaExpansion$ accepts $u$, then 
\begin{itemize}
	\item If $u < \frac{3t}{4}$ then return $w_i$.
	\item Else if $ \frac{3t}{4} \le u < t$:
		\begin{itemize}
			\item Use $\Dec^w_{=t}(w,j)$ for all $ 8tm\ell(\sec) - 4m\ell(\sec) \le j < 8tm\ell(\sec)$ to reconstruct the last block $w_{2t}$. 
			\item If $G_g$ has $4\delta$-local expansion around $t$ using the constructed labeling, then return $w_i$, else return $\bot$.
		\end{itemize}
\end{itemize}
\item If $\DeltaExpansion$ rejects $u$, then return $\bot$.
\end{enumerate}
\end{mdframed}

\begin{lemma}
\lemlab{lem:maj:front}
Let $\Enc$ be as described in \secref{sec:maj:lccencode}.
For any $i \in [n]$ such that $\metanode{i}<t$, $\Dec_{< t}$ does the following:
\begin{enumerate}
\item
For any $x \in \{0,1\}^k$ and $c =\Enc(s, x)$,  $\Dec^c_{< t}(s, i) = c[i]$.
\item
For any $x \in \{0,1\}^k$, $c =\Enc(s, x)$  and $w \in \{0,1\}^n$ generated by any \PPT adversary such that $\HAM(c, w) \le \frac{\Delta_J k}{\error}$, 
\[\PPr{\Dec^w_{<  t}(s, i) \in \{c[i], \bot \} } \ge 1-\negl(n).\]
\end{enumerate}
Moreover, $\Dec^w_{< t}$ makes at most $\O{\ell(\sec) \log^{3+\eps} n}$ queries to input $w$.
\end{lemma}

Before we prove the correctness of \lemref{lem:maj:front}, we show that any \PPT adversary that generates a corrupted codeword $w$ such that $\HAM(w, c) \leq \frac{\Delta_j k}{\error}$, cannot manage to corrupt many meta-nodes. 
Let $T$ denote the set of tampered meta-nodes of $G_0$.

\begin{lemma}\lemlab{lem:tampered:nodes}
Let $w \in \{ 0, 1\}^n$ be a corrupted codeword generated by any $\PPT$ adversary 
such that $\HAM(w, C)\leq \frac{\Delta_j k}{\error}$, then at most  $\frac{t}{4\error}$ meta-nodes are tampered. 
\end{lemma}
\begin{proof}
We first observe that a meta-node can only be altered beyond repair by changing at least $4m \cdot \ell(\sec) \cdot \Delta_J$ bits.  Therefore, by changing at most $\frac{\Delta_J k}{\error}$ bits of the codeword we have $|T| \leq \frac{k}{4 \error m \ell(\sec)} = \frac{k'}{4 \error m} = \frac{t}{4\error}$ meta-nodes in $G$ are  tampered.
\end{proof}

Using these key lemmas and properties, we now prove the correctness of the decoder. 

\begin{proof}[Proof of \lemref{lem:maj:front}]

{\bf Query Complexity: } The query complexity of the decoder is dominated by the query complexity of $\DeltaExpansion$. From \lemref{lem:green:expander}, it then follows that the query complexity of $\Dec^w_{<t}$ is at most $\O{m \ell(\sec) \log^{2+\eps} t} = \O{m \ell(\sec) \log^{2+\eps} n}$ for $m = \O{\log n}$.

{\bf Correctness (1): } If $w \in C$, then the green graph $G_g = G_0$ and hence, $G_g$ has $\delta$ expansion around all meta-nodes $u$. Also, since $T=\emptyset$, all meta-nodes are $\alpha$-good in $G_0$ with respect to $T$. Therefore, $\Dec_{< t}$ accepts a codeword with probability $1$. 

{\bf Correctness (2): } Now, let $w$ be a corrupted codeword such that $\HAM(d, c) \leq  \frac{\Delta_J k}{\error}$. Let $u =\metanode{i}$ be the meta-node corresponding to the queried index $i \in [n]$. 

Case (i):  Let $u < \frac{3t}{4}$, and let $G_g$ have $4\delta$-local expansion around $u$. We show that there exists a descendent $v$ of $u$ in $G_g$ such that (1)  $G_g$ is $4\delta$-local expander around $v$ and, (2) $v$ is untampered (correct).

Recall that $T$ denotes the set of tampered meta-nodes. 
From \lemref{lem:tampered:nodes}, we know that at most $\frac{t}{4\error}$ meta-nodes can be tampered by the adversary.  
Since $G_0$ is a $\delta$-local expander, from \lemref{lem:cgood:most} we know that the number of $\alpha$-good meta-nodes in $G_0$ with respect to the set of tampered meta-nodes $T$ is at least $t - |T|\left( \frac{2-\alpha}{\alpha}\right)$. 
Therefore, for any $\alpha > \frac{1}{200 d_{\delta}} $ and $\error \ge 1600 d_{\delta}$, there are at least $ \frac{15t}{16}$ $\alpha$-good meta-nodes in $G_0$ w.r.t $T$. \lemref{lem:alpha:expansion} then implies that $G_g$ has $4\delta$-local expansion around all these $\alpha$ good meta-nodes.  
So at most $\frac{t}{16} + \frac{t}{16} = \frac{t}{8}$ meta-nodes which do not satisfy conditions $(1)$ or $(2)$. 
Therefore, there exists at least one among the last $\frac{t}{4}$meta-nodes which satisfy both the conditions.

Case (ii): Let $u \geq \frac{3t}{4}$, and let $t$ be the corrected meta-node using $\Dec^w_{=t}$ (see \lemref{lem:maj:back}). If $G_g$ has $4\delta$-local expansion around $u$ and $t$, then assuming that the adversary has not found a hash collision, we can conclude from \lemref{lem:connected:correct}, that $u$ is untampered. So $\Dec^w_{<t}$ returns the correct value for all such coordinate queries. 

The probability that the decoder returns a wrong value is upper bounded by the probability that the procedure $\DeltaExpansion$ wrongly accepts a meta-node about which $G_g$ does not have $4\delta$ local expansion or if the adversary successfully finds a hash collision. This happens with $\negl(n)$ probability as shown in \lemref{lem:green:expander}.
\end{proof}

To conclude, we need to show that for any received word generated by a \PPT adversary most meta-nodes are corrected by the decoder (i.e $\Dec^w(s,i) \neq \bot$ ). In particular, we show that most meta-nodes of $G_0$ are untampered, and  $G_g$ had $4\delta$ local expansion around each of these untampered meta-nodes. 

This follows from the fact that there are at most $\frac{t}{16}$ meta-nodes are not $\alpha$-good under the set of tampered nodes. Note that for any query to a meta-node $u < \frac {3t}{4}$ the decoder $\Dec^w(s,i)$ returns the correct codeword symbol, \ie, $c[i]$ unless 
$u$ itself is not $\alpha$-good under the set of tampered nodes. So, for at least $\frac{3t}{4} - \frac{t}{16} = \frac{11t}{16}$ meta-node queries,  $\Dec^w(s,i)$ does not return a $\bot$. 

%% file: app-rldc.tex
\section{Computationally Relaxed Locally Decodable Codes ($\CRLDC$)}  \applab{crldc}
In this section of the appendix we present the formal definition of a Computationally Relaxed Locally Decodable Codes  ($\CRLDC$) --- defined informally in the main body of the paper.  
We first tweak the definition of a {\em local code} so that local decoding algorithm takes as input an index from the range $i \in [k]$ instead of $i \in [n]$ i.e., the goal is to decode a bit of the original message as opposed to a bit of the original codeword.

\begin{definition} A {\em local code} is a tuple $(\Gen, \Enc, \Dec)$ of probabilistic algorithms such that  
\begin{itemize}
\item $\Gen(1^{\sec})$ takes as input security parameter $\sec$ and generates a public seed $s \in \{0,1\}^{*}$. 
This public seed $s$ is {\em fixed} once and for all.  
\item $\Enc$ takes as input the public seed $s$ and a message $x \in \Sigma^k$ and outputs a codeword $c =\Enc(s,x)$ with $c \in \Sigma^n$.  
\item $\Dec$ takes as input the public seed $s$, an index $i \in [k]$, and is given oracle access to a word $w \in \Sigma^n$. $\Dec^w(s,i)$ outputs a symbol $b\in \Sigma$ (which is supposed to be the value at position $i$ of the original message $x$ i.e., the string $x$ s.t. $\Enc(s,x)$ is closest codeword to $w$).
\end{itemize}
We say that the (information) {\em rate} of the code  $(\Gen, \Enc, \Dec)$ is $k/n$. We say that the code is {\em efficient} if $\Gen,\Enc,\Dec$ are all probabilistic polynomial time (PPT) algorithms.
\end{definition}

Similarly, the notional of a computational adversarial channel is almost identical except that the channel challenges the $\Dec$ with an index $i \in [k]$ (as opposed to $i \in [n]$) and the decoder is supposed to output the $i$th bit of the original message. 
Apart from this change \defref{def:crldc:channel} and \defref{def:crlcc:channel} are identical.

\begin{definition}
\deflab{def:crldc:channel}
 A {\em computational adversarial channel} $\A$ with error rate $\tau$ is an algorithm that interacts with a local code $(\Gen, \Enc, \Dec)$  in rounds, as follows. In each round of the execution, given a security parameter $\sec$, 
\begin{enumerate}
\item  Generate $s\leftarrow\Gen(1^\sec)$; $s$ is public, so $\Enc$, $\Dec$, and $\A$ have access to $s$
\item The channel $\A$ on input $s$ hands a message $x$ to the sender.
\item The sender computes $c=\Enc(s, x)$ and hands it back to the channel (in fact the channel can compute $c$ without this interaction).
\item The channel $\A$ corrupts at most $\tau n$ entries of $c$ to obtain a word $w\in \Sigma^n$ and selects a challenge index $i \in [k]$;  $w$ is given to the receiver's \Dec with query access along with the challenge index $i$.
\item The receiver outputs $b\leftarrow \Dec^w(s, i)$.
\item We define $\A(s)$'s  {\em probability of fooling} $\Dec$ on this round to be $p_{\A, s}=\Pr[b\not \in \{\bot, x_i \}],$ where the probability is taken only over the randomness of the $\Dec^w(s, i)$.
 We say that $\A(s)$ is $\gamma$-successful {\em at fooling} $\Dec$ if $p_{\A, s}>\gamma$.
 We say that $\A(s)$ is $\rho$-successful {\em at limiting} $\Dec$ if $\left|\Goody_{\A,s}\right| < \rho \cdot k$, where $\Goody_{\A,s} \subseteq [k]$ is the set of indices $j$ such that $\Pr[\Dec^w(s,j) = x_j ] > \frac{2}{3}$. We use $\Fool_{\A,s}(\gamma,\tau,\sec)$ (resp. $\Limit_{\A,s}(\rho,\tau,\sec)$) to denote the event that the attacker was $\gamma$-successful at fooling $\Dec$ (resp. $\rho$-successful at limiting $\Dec$) on this round. 
\end{enumerate}
\end{definition}

\begin{definition}[(Computational) Relaxed Locally Decodable Codes (CRLDC)] A local code $(\Gen, \Enc, \Dec)$ is a $(q, \tau, \rho, \gamma(\cdot),\mu(\cdot))$-\CRLCC ~ {\em against a class} $\mathbb{A}$ of adversaries if $\Dec^w$ makes at most $q$ queries to $w$ and satisfies the following:
\begin{enumerate}
\item\label{RLDCweak1} For all public seeds $s$ if $w \leftarrow \Enc(s, x)$ then $\Dec^w(s,i)$ outputs $x_i$. 
\item\label{RLDCweak2} For all $\A \in \mathbb{A}$ we have $\Pr[\Fool_{\A,s}(\gamma(\sec),\tau,\sec)] \leq \mu(\sec)$, where the randomness is taken over the selection of $s \leftarrow \Gen(1^\sec)$ as well as $\A$'s random coins. 
\item \label{RLDCstrong} For all $\A \in \mathbb{A}$ we have $\Pr[\Limit_{\A,s}(\rho,\tau,\sec)] \leq \mu(\sec)$, where the randomness is taken over the selection of $s \leftarrow \Gen(1^\sec)$ as well as $\A$'s random coins. 
\end{enumerate}

When $\mu(\sec)=0$, $\gamma(\sec)=\frac{1}{3}$ is a constant and $\mathbb{A}$ is the set of all (computationally unbounded) channels we say that the code is a $(q, \tau, \rho, \gamma)$-\RLDC. 
When $\mu(\cdot)$ is a negligible function {\em and} $\mathbb{A}$ is restricted to the set of all probabilistic polynomial time (\PPT) attackers we say that the code is a $(q, \tau, \rho, \gamma)$-\CRLDC (computational relaxed locally correctable code). 
  
We say that a code that satisfies conditions \ref{RLDCweak1} and \ref{RLDCweak2} is a {\em Weak \CRLDC}, while a code satisfying conditions \ref{RLDCweak1}, \ref{RLDCweak2} and \ref{RLDCstrong} is a {\em Strong \CRLDC} code.
\end{definition}

As we remarked in the main body of the paper our construction of a Strong {$\CRLCC$} is also a strong $\CRLDC$. In particular, \thmref{strongRLDC} is identical to \thmref{strongRLCC} except that we replaced the word $\CRLCC$ with $\CRLDC$. 

\begin{theorem}\thmlab{strongRLDC}
 Assuming the existence of a collision-resistant hash function $ (\GenH,H)$   with length $\ell(\lambda)$, there exists a constant $0<\tau' < 1$ and negligible functions $\mu(\cdot),\gamma(\cdot)$ such that for all $\tau \leq \tau'$ there exist constants $0< r(\tau), \rho(\tau)<1$ such that there exists a $(\ell(\lambda)\cdot{\polylog n}, \tau, \rho(\tau), \gamma(\cdot),\mu(\cdot))$-Strong \CRLDC of blocklength $n$ over the binary alphabet with rate $r(\tau)$ where $r(.)$ and $\rho(.)$ have the property that $\lim_{\tau \rightarrow 0} r(\tau) = 1$ and $\lim_{\tau \rightarrow 0} \rho(\tau) = 1$. In particular, if $\ell(\lambda) =\polylog \sec$ and $\sec \in \Theta(n)$ then the code is a $({\polylog n}, \tau, \rho, \gamma,\mu(\cdot))$-Strong \CRLDC.
\end{theorem}
\begin{proof}(sketch)
The encoding algorithm in our strong $\CRLDC$ and strong $\CRLCC$ constructions are identical. 
The only change that we need to make is to {\em tweak} the local decoding algorithm to output bits of the original message instead of bits of the codeword. 
This task turns out to be trivial. 
In particular, the first part of the codeword in our construction is formed by partitioning the original message $x$ into $mt$ blocks $x= x_1 \circ \ldots \circ x_{tm}$, partitioning these blocks into $t$ groups $T_1 = \left(x_1 \circ \ldots x_m\right), \ldots, T_t = \left(x_{(t-1)m+1} \circ \ldots x_{tm}\right)$ and outputting $c_j=\ECC(T_j)$ for each $j \leq t$. 
Because our rate $r(\tau)$ approaches $1$ these bits account for {\em almost all} of the codeword.

The $\CRLCC$ decoding algorithm is given (possibly tampered) codeword $\tilde{c} = \tilde{c}_1 \circ \tilde{c}_2 \ldots$. 
When the decoding algorithm is challenged for one of the original bits of $c_j$ it works as follows: (1) use our randomized algorithm verify that the corresponding metanode $v_j$ in our graph is $\alpha$-good (if not we output $\bot$), (2) if $v_j$ is $\alpha$-good then run $\ECCD\left( \tilde{c}_j\right)$ to recover the {\em original} $T_j = \left(x_{(j-1)m+1},\ldots,x_{jm}\right)$, (3) run $\ECC(T_j)$ to recover the original $c_j$ and then find the appropriate bit of the codeword to output.  
The $\CRLDC$ decoding algorithm can simply omit step (3). 
Once we have recovered the original $T_j$ we can find the appropriate bit of the original message to output.
\end{proof}
 